\documentclass[lettersize,journal]{IEEEtran} 

\makeatletter

\let\proof\@undefined
\let\endproof\@undefined
\let\IEEElabelindent\labelindent
\makeatother

\usepackage{graphicx}
\usepackage{enumerate}
\usepackage{hyperref}
\usepackage{makecell}
\usepackage{multirow}
\usepackage{booktabs}
\usepackage{cite}
\usepackage{amsmath} 
\usepackage{amssymb} 
\usepackage{amsthm}  
\usepackage{bm}  
\usepackage{lipsum}
\usepackage[ruled, linesnumbered]{algorithm2e}
\usepackage{color}
\usepackage{algpseudocode}%
\usepackage{mathtools}
\usepackage{textcomp}
\usepackage{verbatim}
\usepackage{breqn}
\usepackage[textsize=tiny]{todonotes}
\usepackage{tikz}
\usepackage{cases}
\usetikzlibrary{shapes,arrows}
\usepackage{nomencl}
\usepackage{bbm}

\makeatletter 
\g@addto@macro{\@algocf@init}{\SetKwInOut{Parameter}{Parameters}}
\makeatother

\newtheorem{proposition}{Proposition}

\newtheorem{assumption}{Assumption}
\newtheorem{lemma}{Lemma}
\newtheorem{theorem}{Theorem}
\newtheorem{remark}{Remark}

\newtheorem{problem}{Problem}

\DeclareMathOperator*{\argmin}{arg\,min}

\def\BigRoman{\uppercase\expandafter{\romannumeral\number\count 255 }}
\def\Romannumeral{\afterassignment\LowRoman\count255=}

\title{Probabilistic Constraint Construction for Network-safe  Load Coordination
\thanks{This work was supported by U.S. National Science Foundation Award CNS-1837680.
The authors are with the Department of Electrical
Engineering and Computer Science, University of Michigan, Ann Arbor,
MI 48109 USA {\tt \{sunhoj,necmiye,jlmath\}@umich.edu}.}
}

\author{Sunho Jang, \IEEEmembership{Student Member, IEEE,} Necmiye Ozay, \IEEEmembership{Senior Member, IEEE,}\\ Johanna L. Mathieu, \IEEEmembership{Senior Member, IEEE}}
\begin{document}

{ }

\maketitle

\begin{abstract}
 Distributed Energy Resources (DERs) can provide balancing services to the grid, but their power variations might cause voltage and current constraint violations in the distribution network, compromising network safety. This could be avoided by including network constraints within DER control formulations, but the entities coordinating DERs (e.g., aggregators) may not have access to network information, which typically is known only to the utility. Therefore, it is challenging to develop network-safe DER control algorithms when the aggregator is not the utility; it requires these entities to coordinate with each other. In this paper, we develop an aggregator-utility coordination framework that enables network-safe control of thermostatically-controlled loads to provide frequency regulation. In our framework, the utility sends a network-safe constraint set on the aggregator's command without directly sharing any network information. We propose a constraint set construction algorithm that guarantees satisfaction of a chance constraint on network safety. Assuming monotonicity of the probability of network safety with respect to the aggregator's command, we leverage the bisection method to find the largest possible constraint set, providing maximum flexibility to the aggregator. Simulations show that, compared to two benchmark algorithms, the proposed approach provides a good balance between service quality and network safety.
\end{abstract}


\makenomenclature

\nomenclature{$n$}{Number of the network's nodes}
\nomenclature{$p_{\text{ref}} (t)$}{Reference signal at time $t$}
\nomenclature{$\mathcal{U} (t)$}{Constraint set on the command at time $t$}
\nomenclature{$p_{\text{agg}} (t)$}{Aggregate power of the TCLs at time $t$}
\nomenclature{$u(t)$}{Probabilistic command from the aggregator to the TCLs}
\nomenclature{$\bm{P} (t)$}{Real power consumptions from the loads}
\nomenclature{$\bm{Q} (t)$}{Reactive power consumptions from the loads}
\nomenclature{$\epsilon$}{Upper bound on the probability of safety violation}
\nomenclature{$\bm{P}^{\text{L}} (t)$}{Real power consumption from the uncontrollable loads}
\nomenclature{$\bm{Q}^{\text{L}} (t)$}{Reactive power consumption from the ncontrollable loads}
\nomenclature{$f_{\bm{P}^{\text{L}},\bm{Q}^{\text{L}}}^t$}{Joint probability density function of $(\bm{P}^{\text{L}} (t), \bm{Q}^{\text{L}} (t))$}
\nomenclature{$\underline{v}$}{Lower bound on the voltage of every node}
\nomenclature{$v_j $}{Voltage at node $j$ of the network}
\nomenclature{$\beta$}{Parameter for the confidence level}
\nomenclature{$n^{\text{TCL}}$}{Number of the TCLs}
\nomenclature{$n_{j}^{\text{TCL}}$}{Number of the TCLs at node $j$}
\nomenclature{$\theta^i (t)$}{Temperature of the $i$th TCL at time $t$}
\nomenclature{$m^i (t)$}{Mode of the $i$th TCL at time $t$}
\nomenclature{$\theta_{\text{a}}^i (t)$}{Ambient temperature of the $i$th TCL}
\nomenclature{$a_{\text{th}}^i$}{A parameter for the temperature evolution dynamics of the $i$th TCL}
\nomenclature{$\Delta t$}{Sampling time}
\nomenclature{$r_{\text{th}}^i$}{Thermal resistance of the $i$th TCL}
\nomenclature{$c_{\text{th}}^i$}{Thermal capacitance of the $i$th TCL}
\nomenclature{$p_{\text{tr}}^i$}{Energy transfer rate of the $i$th TCL}
\nomenclature{$p^i$}{Rated power consumption of all the TCLs}
\nomenclature{$\zeta^i$}{Coefficient of performance of the $i$th TCL}
\nomenclature{$q^i$}{Rated reactive power consumption}
\nomenclature{$\phi^i$}{Power factor of the TCLs}
\nomenclature{$\underline{\theta}^i$}{Lower bound of the temperature dead-band of $i$th TCL}
\nomenclature{$\overline{\theta}^i$}{Upper bound of the temperature dead-band of $i$th TCL}
\nomenclature{$\theta_{\text{s}}^i$}{Set point temperature}
\nomenclature{$\gamma^i$}{Width of the temperature dead-band}
\nomenclature{$z^i (t)$}{Random number sampled by the $i$th TCL for mode switching}
\nomenclature{$\bm{P}^{\text{T}} (t)$}{Real power consumption from the TCLs}
\nomenclature{$\bm{Q}^{\text{T}} (t)$}{Reactive power consumption from the TCLs}
\nomenclature{$\bm{N}^{\text{ON}} (t)$}{Number of the TCLs in ON mode}
\nomenclature{$\bm{N}^{\text{OFF}} (t)$}{Number of the TCLs in OFF mode}
\nomenclature{$\bm{S}^{\text{ON}} (t)$}{Number of the TCLs internally switched ON}
\nomenclature{$\bm{S}^{\text{OFF}} (t)$}{Number of the TCLs internally switched OFF}
\nomenclature{$w^{\text{ON}} (t)$}{The portion of the TCLs internally switched ON}
\nomenclature{$w^{\text{OFF}} (t)$}{The portion of internally s switched TCLs to OFF mode}
\nomenclature{$\bm{C}^{\text{ON}} (t)$}{Number of the TCLs switched ON by the aggregator's command}
\nomenclature{$\bm{C}^{\text{OFF}} (t)$}{Number of the TCLs switched OFF by the aggregator's command}
\nomenclature{$\bm{P}_{u} (t+1)$}{Real power at time $t+1$ under the command $u(t+1) = u$}
\nomenclature{$\bm{Q}_{u} (t+1)$}{Reactive power at time $t+1$ under the command $u(t+1) = u$}
\nomenclature{$p_{j}^{\text{b}}$}{Real power flowing at the branch whose receiving end is node $j$}
\nomenclature{$q_{j}^{\text{b}}$}{Reactive power flowing at the branch whose receiving end is node $j$}
\nomenclature{$r_{j}$}{Resistance of the branch whose receiving end is node $j$}
\nomenclature{$x_{j}$}{Reactance of the branch whose receiving end is node $j$}
\nomenclature{$e (j)$}{Parent node of node $j$}
\nomenclature{$c (j)$}{Set of child nodes of node $j$}
\nomenclature{$f_{v_j}$}{Solution of the branch flow equations corresponding to the voltage at node $j$}
\nomenclature{$\bm{V}_{u} (t+1)$}{Voltage at each node under $u (t+1) =u$}
\nomenclature{$X_{u} (t+1)$}{Network safety under $u(t+1) = u$}
\nomenclature{$\nu_{u} (t+1) $}{Probability of network safety under $u(t+1) = u$ at time $t+1$}
\nomenclature{$\mathbb{N}$}{Set of natural numbers}
\nomenclature{$[N]$}{Set of natural numbers from 1 to $N$}
\nomenclature{$[N]_{0}$}{Set of natural numbers from 0 to $N$}
\nomenclature{$\mathcal{B} ( n_{\text{s}}, \nu )$}{The binomial distribution with the number of trials $n_{\text{s}}$ and the success probability $\nu$}
\nomenclature{$\mathcal{F}_{\text{B}} ( x ; n_{\text{s}} , \nu )$}{Cumulative Density function of $\mathcal{B} ( n_{\text{s}}, \nu )$}
\nomenclature{$\mathcal{N} ( \mu, \sigma)$}{The normal distribution with mean $\mu$ and variance $\sigma$}
\nomenclature{$m_{c}$}{Mode switch function by aggregator's command}
\nomenclature{$v_0$}{Voltage at the substation.}
\nomenclature{$\hat{f}_{p_{j}^{\text{b}}}$}{Solution of simplified Distflow equations corresponding to the branch real power flow $p_j^{\text{b}}$}
\nomenclature{$\hat{f}_{q_{j}^{\text{b}}}$}{Solution of simplified Distflow equations corresponding to the branch reactive power flow $q_j^{\text{b}}$}
\nomenclature{$p_{j}^{\text{Ln}}$}{Nominal real power consumption  of the uncontrollable loads at node $j$}
\nomenclature{$q_{j}^{\text{Ln}}$}{Nominal reactive power consumption  of the uncontrollable loads at node $j$}
\nomenclature{$\bm{A}(u(t))$}{Transition matrix of bin-model under command $u(t)$}
\nomenclature{$\overline{p}_j$}{Average real power of the TCLs at node $j$}
\nomenclature{$\overline{q}_{j}$}{Average reactive power of the TCLs at node $j$}


\section{Introduction}
\IEEEPARstart{A}{s} the amount of intermittent renewable generation is rapidly growing, it is becoming more difficult to rely solely on the conventional ways of balancing power systems. One emerging solution is to leverage  Distributed Energy Resources (DERs), such as thermostatically-controlled loads (TCLs), batteries, and electric vehicles, to provide grid services. By doing so, they can improve the reliability, and reduce the operating cost and environmental impact of power systems. However, DERs coordinated to provide balancing services might cause issues in the distribution network, such as under/over-voltages, over-current violations, and transformer overheating, compromising network safety.
 
When the distribution network operator (i.e., the utility) coordinates DERs to provide grid services it can adopt a centralized algorithm that explicitly manages distribution network constraints, e.g., the algorithms provided in \cite{dall2017optimal,bernstein2019real,vrettos2013combined}. However, in competitive U.S. electricity markets it is becoming more likely that third-party (i.e., non-utility) DER aggregators will take on this role.  Unfortunately, the aggregator does not have access to detailed distribution network information typically known only to the utility, and so it is unable to directly determine how its actions would affect the distribution network. This challenge has already been recognized by the US Federal Energy Regulatory Commission (FERC)~\cite{ferc2018notc}.
 
Thus, there is a need for coordination between the aggregator and the utility to ensure network-safe operation of DERs. The recent FERC Order No. 2222~\cite{ferc2222} provided some guidance on the development of operational coordination architectures between DER aggregators, utilities, and market coordinators; however, it is still unclear how these architectures will evolve and which architecture is ``best." Beyond ensuring network safety, coordination architectures should also 1) ensure that each entity's private information (e.g., sensitive network information held by the utility, proprietary DER coordination strategies held by the aggregator, and private DER state information held by the DERs' end-users) is not shared with the other entities and 2) communication between the entities is minimal for compatibility with current communications infrastructure and/or to reduce the cost of any newly required infrastructure. Furthermore, architectures need to specify coordination protocols on different timescales, for example, 1) for operational planning such that the aggregator can determine its offer for balancing services, and 2) for real-time control in case network conditions differ significantly from forecasts and aggregator actions need to be curtailed.

In this paper, we propose an aggregator-utility coordination framework for a collection of TCLs to provide balancing services like frequency regulation while ensuring distribution network-safety with high probability. We focus on real-time coordination, specifically, a setting in which an aggregator is already committed to provide a certain amount of balancing services, but real-time distribution network conditions require curtailment of those services. In our framework, the utility sends the aggregator a one-step ahead constraint set on the aggregator's control input, which guarantees the satisfaction of a chance constraint on network safety with a certain confidence level. This method leverages estimation from Monte Carlo simulation and the bisection method to provide the largest possible constraint set to maximize the network-safe TCL flexibility. To achieve light communication requirements, the aggregator control algorithm assumes the TCLs all respond to the same scalar control input. This constrains the aggregator's degrees-of-freedom but also makes it possible for the utility to define a simple constraint set on the control input.

Previous work, e.g.~\cite{mathieu2012state,bashash2012modeling, zhang2013aggregated,tindemans2015decentralized}, has proposed strategies to control aggregations of TCLs, such as air conditioners and water heaters, to provide balancing services in ways that are non-disruptive to end-users.  TCLs have inherent thermal energy storage capacity and non-disruptiveness can be achieved, e.g., by keeping internal temperatures inside a narrow temperature dead-band. However, network safety was not considered in the above papers. Some work has developed network-safe control algorithms for TCLs coordinated by third-party aggregators. Ref.~\cite{ross2019coordination} proposes both a utility-centric and an aggregator-centric coordination framework, differentiated by which entity ultimately sends control commands to the TCLs. That paper and~\cite{ross2021strategies} develop utility-centric strategies wherein the utility blocks aggregator's commands that would otherwise cause network constraint violations. In contrast, our proposed approach would be considered aggregator-centric.

Aggregator-centric network-safe DER coordination could be achieved through (convex) inner approximation of safe operating regions~\cite{lee2021robust,nguyen2018constructing,nazir2021grid}, which could be computed by the utility and sent to the aggregator as constraints on the net DER power deviations at each node. Research from Australia refers to these nodal constraints as operating envelopes \cite{petrou2020operating,yi2022fair,russell2022stochastic}. Ref.~\cite{comden2022secure} proposes an optimization problem to obtain a hyper-rectangular constraint set on the net power consumption of controllable DERs at each node in order to satisfy chance constraints on the voltage at each node. However, these approaches all require constraints to be applied at each node, rather than applying a constraint on aggregate power deviations by DERs located across a network. Ref.~\cite{ross2020method} proposes a method to constrain the norm of the power deviations across all nodes in the network, but requires significant computation to compute the constraint. Assuming an aggregate power deviation constraint exists, our previous work~\cite{jang2021large} develops an aggregator-centric TCL coordination algorithm using formal methods, but does not develop an approach to obtain the constraint, and the solutions are very conservative. 

In contrast to this previous work, this paper makes the following contributions: 1) we develop a new aggregator-centric approach to enable network-safe control of TCLs for balancing services; 2) assuming a simple control scheme that leverages a scalar control input to coordinate TCLs to provide balancing services (the aggregator's algorithm), we develop an approach to constrain the control input to satisfy a chance constraint on network safety (the utility's algorithm); and 3)~we demonstrate our approach in simulation and compare its performance to two benchmark approaches. In contrast to past work on network-safe control that assumes the system is deterministic, e.g.,~\cite{ross2020method}, here we consider uncertainty in the power consumption of non-participating loads. Furthermore, in contrast to~\cite{jang2021large}, we assume the aggregator has incomplete information about the TCLs to reduce communication requirements and preserve some level of privacy. Lastly, though
some past work leveraged chance constraints to develop network-safe DER coordination approaches, e.g.,~\cite{baker2016distribution,dall2017chance,hassan2018chance,ayyagari2017chance,hassan2018optimal,chen2021combining,comden2022secure}, these papers all assume that the controller has detailed distribution network information (enabling the formulation of a chance-constrained optimal power flow problem), which is inconsistent with our utility-aggregator coordination framework.
 
The organization of the paper is as follows. Section~\ref{sec:framework_problem} introduces the coordination framework and problem of interest. Section~\ref{sec:aggregator_TCLs} explains the aggregator's control approach and Section~\ref{sec:utility} details the proposed constraint construction algorithm used by the utility to achieve network safety at a high level of probability. Section~\ref{sec:numerical_experiment} presents the results of a case study comparing the proposed approach to two benchmarks. The appendix includes proofs of two of the theorems.

\emph{Notation:}
$\mathbb{N}$, $[N]$, $[N]_0$ are the set of natural numbers, $\{ 1,\ldots,N \}$, and $\{ 0, 1, \ldots, N \}$, respectively. The $n$-dimensional Euclidean space is $\mathbb{R}^n$. The $j$th element of the vector $\bm{y}$ is $y_{j}$.
Binomial distribution $\mathcal{B} (n_{\text{s}}, \nu)$ has $n_{\text{s}}$ trials, each with success probability $\nu$, and cumulative density function (cdf) $\mathcal{F}_{\text{B}} (x ; n_{\text{s}},\nu)$. $\mathcal{N} (\mu, \sigma^2 )$ is the normal distribution with mean $\mu$ and variance $\sigma^2$. Function $\mathbbm{1}(A)$ is 1 if $A$ is true, and 0 otherwise. All random variables are capitalized English letters, e.g., $X$, with realizations denoted $\tilde{x}$ and estimates/approximates denoted $\hat{x}$. All other variables are denoted by symbols other than capitalized English letters. Vectors and matrices are bolded.

\section{Framework \& Problem of interest} \label{sec:framework_problem}

We consider a framework in which a utility and aggregator coordinate to provide network-safe grid balancing services, e.g., frequency regulation, by aggregations of TCLs. TCLs switch ON/OFF to maintain temperature within a dead-band. We focus on real-time coordination, i.e., we assume that the aggregator has already participated in the ancillary services market and committed balancing service capacity to the independent system operator (ISO). The amount of balancing service capacity offered by the aggregator was based on forecasts of the capabilities of the TCLs and the network state. However, the real-time network state differs significantly from its forecasts and so the committed balancing service capacity must be curtailed to avoid distribution network constraint violations. This could happen when load consumption and/or renewable power injections are significantly different from forecasts and the network is operating close to its limits.

We assume that the following coordination steps occur at each discrete time step $t$, where the length of each time step is $\Delta t$. The coordination scheme is shown in Fig.~\ref{fig:entities_relationship}.
    \begin{enumerate}
        \item The aggregator receives a constraint set $\mathcal{U} (t)$ from the utility and a reference signal $p_{\text{ref}} (t)$ (e.g., a scaled and biased frequency regulation signal) from the ISO.
        \item The aggregator determines the control command $u(t) \in \mathcal{U}(t)$ and broadcasts the same command to all TCLs.
        \item Each TCL maintains or switches its ON/OFF mode based on its temperature and the aggregator's command $u(t)$.
        \item The utility observes the real and reactive power consumption at each network node $\bm{p} (t)$ and $\bm{q} (t)$, and obtains some information from the aggregator (described below). Then, it constructs a one-step ahead constraint set $ \mathcal{U}(t+1) $ and sends it to the aggregator. (And go back to step 1.)
    \end{enumerate}
 
    \begin{figure}[t]
        \centering
        \includegraphics[width=0.95\linewidth]{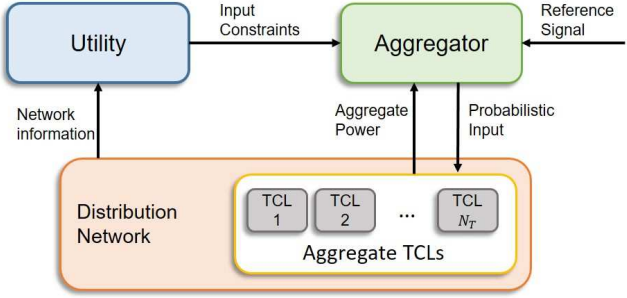}
        \caption{Coordination between the aggregator, utility, and the TCLs.}\vspace{-0.5cm}
        \label{fig:entities_relationship}
    \end{figure}
    
The aggregator's goal is to select $u(t)$ to maximize the quality of grid balancing services. This means that the aggregator should choose a command $u(t)$ that is likely to adjust the aggregate power of the TCLs to match the reference signal $p_{\text{ref}} ( t )$ as closely as possible. Here, we assume the aggregator's command $u(t)$ is a real scalar in the range $ [ -1 ,1 ] $ and is interpreted by each TCL as the probability it should switch modes; the details of how it switches are given in Section~\ref{sec:aggregator_TCLs}. TCL coordination through probabilistic switching has been considered in previous work e.g.,~\cite{mathieu2012state,ross2019coordination}. An advantage of this type of command is that it only needs simple broadcast communication infrastructure. However, it does not allow the aggregator to directly adjust the power consumption of individual TCLs, which means that the aggregator has a low degree-of-freedom in control.
    
Since the aggregator does not have detailed distribution network information and cannot evaluate how its command would affect the network, the utility sends a one-step ahead constraint set $\mathcal{U} (t+1)$ on the aggregator's command $u(t+1)$. This set $\mathcal{U} (t+1)$ is designed such that, if $u(t+1) \in \mathcal{U} (t+1)$, then probability of network safety is over a desired value $1-\epsilon$. We propose a method for the utility to compute $\mathcal{U} (t+1)$ in Section~\ref{sec:utility}, which is the main contribution of this work. To do this, the utility leverages:

1) Real-time data from household smart meters to obtain the real and reactive power consumption at each node,  $\bm{p} (t)$ and $\bm{q} (t)$. We recognize that in practice most utilities do not currently gather smart meter data in real-time, but this is possible with most existing smart meters and could be enabled through reconfiguration of their settings.

2) Forecasts of the probability distributions of the one-step ahead real and reactive power consumption of non-participating loads at each node, $\bm{P}^{\text{L}} (t+1)$ and $\bm{Q}^{\text{L}} (t+1)$. We assume that these distributions are estimated using historical and real-time data from household smart meters, and leveraging a disaggregation technique~\cite{hart1992nonintrusive} to separate the power consumption of the TCLs from that of the non-participating loads. We assume that $ \bm{P}^{\text{L}} (t) $ and $\bm{Q}^{\text{L}} (t)$ are correlated and $f_{\bm{P}^{\text{L}}, \bm{Q}^{\text{L}}}(t)$ is their joint probability density function (pdf).

3) Some real-time TCL information from the aggregator that is necessary for constraint set computation. This should be minimal to protect end-user privacy. In our framework, the aggregator provides the one-step ahead estimated fractions of TCLs that will be outside of their temperature dead-band and switched OFF-to-ON and ON-to-OFF by their thermostats, $\bm{\hat{w}}^{\text{ON}} (t+1)$ and $\bm{\hat{w}}^{\text{OFF}} (t+1)$. Details on how this information is used are provided in Section~\ref{sec:probsafety}. 

In this paper, for the sake of simplicity, we define network safety in terms of under-voltage violations. Specifically, we say that the network is safe if there are no under-voltage violations, and unsafe if there are any violations. The approach can be easily extended to include over-voltage violations and other distribution network constraint violations. The formal statement problem is as follows.
    \begin{problem} \label{prob:problem_statement}
     Given the desired safety probability $1-\epsilon$, the real-time real and reactive power consumption at each node $\bm{p}(t)$ and $\bm{q} (t)$, the joint pdfs of the uncontrollable loads $f_{\bm{P}^{\mathrm{L}},\bm{Q}^{\mathrm{L}}}(t)$, $f_{\bm{P}^{\mathrm{L}},\bm{Q}^{\mathrm{L}}}(t+1)$, and the fractions of TCLs that are outside of their dead-band $\bm{w}^{\mathrm{ON}} (t+1)$, $\bm{w}^{\mathrm{OFF}} (t+1)$, find a one-step ahead constraint set $\mathcal{U} (t+1)$ such that the following chance constraint holds if $u(t+1) \in \mathcal{U} (t+1)$,
       \begin{equation} \label{eq:chance_constrained_constraint}
           \mathrm{Pr} \left( \min_{j \in [n]} V_{j} (t+1) \geq \underline{v} \right ) \geq 1 - \epsilon,
       \end{equation}
       where $\underline{v}$ is the lower bound on each of the nodal voltages $V_j$  and $n$ is the number of nodes other than the substation.
    \end{problem}
    To solve this problem, we define the one-time step ahead voltage at each node $ V_{j} ( t+1 )$ as a random variable whose distribution depends on the command $u (t+1)$; the details are explained in Section~\ref{sec:utility}. It is difficult to obtain a closed-form expression for the probability distribution of each $V_{j} (t+1) $. Therefore, our approach leverages Monte Carlo simulation to estimate the left side of \eqref{eq:chance_constrained_constraint} given a one-step ahead command $u (t+1)$. Since estimation from sampling leads to error, we find a constraint set $\mathcal{U} (t+1)$ with a confidence level over a desired level $1 - \beta$ rather than giving an exact solution.
    
    \section{Aggregator's Control Approach} \label{sec:aggregator_TCLs}
    
    In this section, we explain how the TCLs operate under the aggregator's command $u(t)$. For simplicity, we assume that all participating TCLs are cooling TCLs (e.g., air conditioners), though the approach also applies to heating TCLs. We denote by $\bm{n}^{\text{TCL}}$ the vector whose element $n_j^{\text{TCL}}$ is the number of participating TCLs at node $j$, and by $n^{\text{TCL}} := \bm{1}^\top \bm{n}^{\text{TCL}}$ the total number of participating TCLs, which satisfies $\sum_{j=1}^{n} n_{j}^{\text{TCL}} = n^{\text{TCL}}$. The internal temperature of the $i$th TCL at time $t$ is denoted by $\theta^{i} (t)$ and its mode is denoted by $m^i (t)$, which is 0 when it is OFF, and 1 when it is ON. The temperature dynamics of the $i$th TCL follow the affine model from \cite{sonderegger1978dynamic}, \begin{equation} \label{eq:TCLtempdyns}
        \theta^{i} (t+1) = a_{\text{th}}^i \theta^i (t) + \left ( 1 - a_{\text{th}}^i \right ) \left ( \theta_a^{i} (t) + r_{\text{th}}^i p_{\text{tr}}^{i} m^{i} (t) \right ),
    \end{equation}
where $\theta_a^i (t) $ is the ambient temperature and $a_{\text{th}}^i = \exp ( - \Delta t / (r_{\text{th}}^i c_{\text{th}}^i))$ is a parameter computed from the thermal resistance $r_{\text{th}}^i$ and capacitance $c_{\text{th}}^i$ of the $i$th TCL. Also, $p_{\text{tr}}^{i}$ is the energy transfer rate of the $i$th TCL, which is negative for a cooling TCL. The power consumption of the $i$th TCL in the ON mode is $p^{i} := p_{\text{tr}}^{i} / \zeta^{i} $ where $\zeta^i$ is the coefficient of performance; the power consumption in the OFF mode is $0$. We assume that the reactive power consumption of the $i$th TCL is $q^i := \omega^i p^i$, where $\omega^i$ is a positive constant.  
The aggregate real power consumption of the TCLs is $p_{\text{agg}} (t) := \sum_{i=1}^{n^{\text{TCL}}} p^i m^i (t)$. 
    
Each TCL has a temperature range $[ \underline{\theta}^i, \overline{\theta}^i ]$ within which its internal temperature should always be; this range is called the temperature dead-band. The temperature set-point, which is set by its end-user, $\theta_{\text{s}}^{i} := ( \underline{\theta}^i + \overline{\theta}^i ) / 2$ is the middle point of the dead-band. Whenever a TCL's internal temperature reaches or goes beyond the boundary of its dead-band it switches its mode to go back into the dead-band.

At each time step $t$, the aggregator determines its command $u(t)$ and broadcasts it to all participating TCLs. TCLs within their dead-bands interpret this command as the desired probability of OFF TCLs to switch ON when $u(t)>0$, and the desired probability of ON TCLs to switch OFF when $u(t)<0$. To determine whether or not to switch, each TCL draws a random number $z^{i} (t)$ from the uniform distribution on the interval $[ 0 , 1 )$ and compares it to the command $u (t)$. If it is OFF and $ z^{i} (t) \leq u (t)$, then it switches ON. If it is ON and $ z^{i} (t) \leq -u (t)$, then it switches OFF. 

In summary, the mode of the $i$th TCL is
    \begin{equation} \label{eq:modechange} 
    m^i (t) = 
        \begin{cases}
            1 & \text{if }  \theta^i (t) \geq \overline{\theta}^i \\
            0 & \text{if }  \theta^i (t) \leq \underline{\theta}^i  \\
            m_{c} ( z^i(t), u(t)) & \text{otherwise,}
        \end{cases}
    \end{equation}
    where $m_c ( z^i (t), u(t) )$ is equal to
    \begin{equation} \label{eq:switchbyinput}
    \begin{cases}
        1 & \text{if } m^i (t-1) = 0 \text{ and } z^i(t) \leq u(t) \\
        0 & \text{if } m^i (t-1) = 1 \text{ and } z^i(t) \leq -u(t) \\
        m^i (t-1) & \text{otherwise}.\notag
    \end{cases}
    \end{equation}
     
     Note that, 
     when positive (negative) $u(t)$ is broadcast to the TCLs, the fraction of the OFF (ON) TCLs within their dead-bands that are switched is approximately $u(t)$ ($-u(t)$). Thus, $|u(t)|$ can be interpreted by the aggregator as the \textit{ratio} of the power consumption increase (decrease) compared to the maximal increase (decrease). Therefore, even though the power consumption of each TCL is not directly controlled by the aggregator, the aggregator can manipulate $p_{\text{agg}} (t)$ by selecting the $u(t) \in \mathcal{U} (t)$ that is likely to adjust $p_{\text{agg}} (t)$ to match the reference signal $ p_{\text{ref}} (t)$ as closely as possible, i.e.,
     \begin{equation}
     \label{eqn:uopt}
     u_{\text{opt}} ( t ) = \argmin_{ u \in \mathcal{U} (t) } \left | \mathbb{E} \left [ P_{\text{agg}} ( t ) \right ]  - p_{\text{ref}} ( t ) \right |,
\end{equation}
where $\mathcal{U} (t)$ is provided by the utility.

 \section{Utility's constraint construction method} \label{sec:utility}

 As mentioned in Section~\ref{sec:framework_problem}, the utility computes a one-step ahead constraint set $\mathcal{U} (t+1)$, which should be a solution to Problem~\ref{prob:problem_statement}. This requires the utility to be able to evaluate how the command $u(t+1)$ would affect the probability of network safety. In this section, 
we first show how the voltage at each node is modeled as a random variable. For ease of exposition, we consider only balanced radial distribution networks. Then, we derive the probability of network safety (i.e., the probability that no under-voltage violations happen) as a function of the command $u (t+1) = u$. 

Next, we show how to verify whether or not the chance constraint~\eqref{eq:chance_constrained_constraint} is satisfied under $u (t+1) = u$ with a desired confidence level, and how the utility can construct $\mathcal{U} (t+1)$ to ensure \eqref{eq:chance_constrained_constraint} is satisfied. We introduce a theorem establishing a confidence interval for the success probability of a Bernoulli random variable using Monte Carlo simulations. Using this result, we leverage the bisection method to find the largest upper bound on $u (t+1)$ that guarantees \eqref{eq:chance_constrained_constraint} with a desired confidence level. The largest upper bound gives the aggregator the greatest possible flexibility in determining its command. 
   
   \subsection{Modeling the probability of network safety}
   \label{sec:probsafety}
   
     We denote the real and reactive power consumption of participating TCLs across all nodes by $\bm{P}^{\text{T}} (t)$ and $\bm{Q}^{\text{T}} (t) \in \mathbb{R}^n$. 
     The utility approximates the nodal values as
    \begin{equation} \label{eq:TCLpowerapprox}
        \begin{aligned}
            P_{j}^{\text{T}} (t) & \approx \overline{p}_{j} N_{j}^{\text{ON}} (t), \enspace Q_{j}^{\text{T}} (t) \approx \overline{q}_{j} N_{j}^{\text{ON}} (t) \quad \forall j \in [n],
        \end{aligned}
    \end{equation}
    where $N_j^{\text{ON}} (t)$ and $N_j^{\text{OFF}} (t)$ are the number of ON and OFF TCLs at node $j$, and $\overline{p}_{j}$ and $\overline{q}_{j}$ are the average real and reactive power rating (i.e., the ON-mode consumption) of the TCLs at node $j$. We additionally define diagonal matrices $\Xi_{p} $ and $\Xi_{q} \in \mathbb{R}^{n \times n}$ whose $j$th diagonal elements are $\overline{p}_j$ and $\overline{q}_j$, respectively. Then, $\bm{P}^{\text{T}} (t)=\Xi_{p} \bm{N}^{\text{ON}} (t)$ and $\bm{Q}^{\text{T}} (t)=\Xi_{q} \bm{N}^{\text{ON}} (t)$, and the total real and reactive power consumption across all nodes is $\bm{P} (t) = \Xi_{p} \bm{N}^{\text{ON}} (t) + \bm{P}^{\text{L}} (t)$ and $\bm{Q} (t) = \Xi_{q} \bm{N}^{\text{ON}} (t) + \bm{Q}^{\text{L}} (t).$
    
 We first show how the one-step ahead number of ON TCLs $\bm{N}^{\text{ON}} (t+1) \in \mathbb{R}^n$ is modeled as a random variable under the command $u (t+1) = u$. The number $\bm{N}^{\text{ON}} (t+1)$ depends upon how many TCLs are switched both by their thermostat (i.e., the first and second cases of \eqref{eq:modechange}) and by the aggregator's command (i.e., the third case of \eqref{eq:modechange}). The number of TCLs at each node $j$ that will be switched ON, OFF by their thermostats is
 \begin{equation} \label{eq:numofIntswitchTCLs}
     \begin{aligned}
         S_{j}^{\text{ON}} ( t+1 ) & =  w_{j}^{\text{ON}} ( t+1 ) N_{j}^{\text{OFF}} ( t),  \\ 
         S_{j}^{\text{OFF}} ( t+1 ) & =  w_{j}^{\text{OFF}} ( t+1 ) N_{j}^{\text{ON}} (t),
     \end{aligned}
 \end{equation} 
     where, as defined in Section~
     \ref{sec:framework_problem}, $w_{j}^{\mathrm{ON}} ( t+1 )$ is the one-step ahead fraction of OFF TCLs that will be switched ON and $w_{j}^{\mathrm{OFF}} ( t+1 )$ is the one-step ahead fraction of ON TCLs that will be switched OFF by their thermostats at bus $j$.
     We assume that the aggregator estimates $ w_{j}^{\text{ON}} ( t + 1 )$ and $w_{j}^{\text{OFF}} ( t + 1 )$ using a model of the aggregate TCL dynamics and sends the estimated values $\hat{w}_{j}^{\text{ON}} ( t + 1 )$ and $\hat{w}_{j}^{\text{OFF}} ( t + 1 )$ to the utility, which corresponds to the TCL information illustrated in Fig.~\ref{fig:entities_relationship}. The utility uses these estimates to obtain realizations of $S_{j}^{\text{ON}} (t+1)$ and $S_{j}^{\text{OFF}} ( t+1 )$ via Monte Carlo simulation, which will be explained in Section~\ref{sec:bound_construction}. 

    According to \eqref{eq:modechange}, the numbers of TCLs at each node $j$ that will be switched ON and OFF by the aggregator's command follow binomial distributions, 
    \begin{equation} \label{eq:numofswitchTCLs}
        \begin{aligned}
             C_{u,j}^{\text{ON}}  (t+1) & \sim \mathcal{B} \left ( N_{j}^{\text{OFF}} (t) - S_{j}^{\text{ON}} (t+1), u^{+} \right ), \\
            C_{u,j}^{\text{OFF}} (t+1) & \sim \mathcal{B} \left ( N_{j}^{\text{ON}} (t) - S_{j}^{\text{OFF}} (t+1), u^{-} \right ),
        \end{aligned}
    \end{equation}
    where $u^{+} := \max ( u,0) $ and $u^{-} = \max (-u,0)$.
    Therefore, the number of ON TCLs given the command $u(t+1) = u $ is
    \begin{equation} \label{eq:nextnOn}
        \begin{aligned}
            \bm{N}_{u}^{\text{ON}} (t+1)  =  \;\bm{N}^{\text{ON}} (t) & +  \bm{S}^{\text{ON}} (t+1)  \\
             - \bm{S}^{\text{OFF}} (t+1) & + \bm{C}_{u}^{\text{ON}} (t+1) - \bm{C}_{u}^{\text{OFF}} (t+1). 
        \end{aligned}
    \end{equation}
     Since the distributions of $C_{u,j}^{\text{ON}}  (t+1)$ and $C_{u,j}^{\text{OFF}} (t+1)$ depend on $u$, the real and reactive power consumption across all nodes $\bm{P} (t+1)$ and $\bm{Q} (t+1)$ also depend on $u$. Therefore, from now on, we denote these random variables under the one-step ahead command $ u (t+1) =  u $ as $\bm{P}_{u} (t+1)$ and $\bm{Q}_{u} (t+1)$.
    
    The next step is to model the one-step ahead voltage $V_{j} (t+1)$ at each node $j$ as a random variable. Suppose that $v_j$ is the voltage magnitude at node $j$; $p_{j}^{\text{b}}$ and $q_{j}^{\text{b}}$ are the real and reactive power flowing through the branch whose receiving end is node~$j$; and the resistance and reactance of the branch are $r_j > 0 $ and $x_j > 0$, respectively. Then, the DistFlow equations~\cite{baran1989optimal} corresponding to a single-phase equivalent model of a radial three-phase balanced network are
    \begin{equation} \label{eq:branch_flow_equations}
        \begin{aligned}
            p_{j}^{\text{b}} = & \sum_{k \in c(j)} p_{k}^{\text{b}} + p_j + r_j | i_j^{\text{b}} | 
            \\
            q_{j}^{\text{b}} = & \sum_{k \in c(j)} q_{k}^{\text{b}} + q_j + x_j | i_j^{\text{b}} | 
            \\
            v_{j}^2  = & \; v_{e(j)}^2 - 2 (r_j p_{j}^{\text{b}} +  x_j q_{j}^{\text{b}}) + (r_j^2 + x_j^2 ) | i_j^{\text{b}} |, 
        \end{aligned}
    \end{equation}
    where $e(j)$ and $c(j)$ are the parent node and set of child nodes of node $j$, respectively, and $| i_j^{\text{b}} | =  ((p_{j}^{\text{b}})^2 + (q_{j}^{\text{b}})^2 ) / v_{e(j)}^2 $ is the magnitude of the current flowing through the branch whose receiving end is node $j$. Given real and reactive power consumption $\bm{p}$ and $\bm{q} \in \mathbb{R}^n$ and substation voltage $v_0$, we let $f_{v_{j}} ( \bm{p},\bm{q}, v_0)$ be the voltage solution of \eqref{eq:branch_flow_equations}, which can be obtained by various algorithms such as Backward-Forward Sweep~\cite{kersting2018distribution}. Then, the one-step ahead voltage at node $j$ under the command $ u(t+1) = u$ is $V_{u,j} (t+1)= f_{v_{j}} (\bm{P}_{u} (t+1), \bm{Q}_{u} (t+1), v_0)$. Note that we cannot obtain an explicit pdf of $V_{u,j} (t+1)$  since there is no closed-form solution of $f_{v_j}$. Instead, we can obtain a realization of $V_{u,j} (t+1)$ by solving \eqref{eq:branch_flow_equations} for a set of realizations $\tilde{\bm{p}}$ and $\tilde{\bm{q}}$ of $\bm{P}_{u} (t+1)$ and $\bm{Q}_{u} (t+1)$.
    
    
    Finally, we define a Bernoulli random variable that indicates whether or not an under-voltage violation exists, 
    \begin{equation} \label{eq:safety_indicator}
        X_{u} (t+1) = \mathbbm{1} \left ( \text{min}_{j \in [n]} V_{u,j} ( t+1) \geq \underline{v} \right ),
    \end{equation}
   whose success probability $
   \nu_{u} ( t+1 ) = \mathrm{Pr} \left ( X_{u} (t+1) =1  \right )$ corresponds to the one-step ahead probability of network safety under command $u(t+1)=u$. Thus, the utility's problem is to find a set $\mathcal{U} (t+1)$ such that, for any $u \in \mathcal{U} (t+1)$, $\nu_{u} (t+1)$ is larger than $1-\epsilon$ with confidence level over $1-\beta$. The solution to this problem is explained in the next section.
    
\subsection{Probabilistically-safe set construction} \label{sec:bound_construction}

In this section, we first present a theorem on computing a confidence interval for the success probability of a Bernoulli random variable via a Monte Carlo simulation. Based on this theorem, we then show how the utility can test whether a command $u(t+1)=u$ is probabilistically safe and how this test procedure can be used to construct the set $\mathcal{U} (t+1)$ of all commands that satisfy the chance constraint.  

\begin{theorem} \label{thm:confidence_interval_condition}
        Suppose that $X^{(1)},\ldots,X^{( n_{\text{s}})}$ are i.i.d. samples of a random variable $X$ following Bernoulli distribution $ B(1,\nu)$ for a positive $\nu$ (i.e. $\mathrm{Pr} (X^{(i)} = 1) = \nu$, $\mathrm{Pr} (X^{(i)} = 0 ) = 1 - \nu$ for any $i \in [ n_{\text{s}} ] $). Let $ M_{ n_{\text{s}} } := \sum_{i=1}^{ n_{\text{s}} } X^{(i)} / n_{\text{s}} $ be the estimator of $\nu$, and $\tilde{m}_{n_{\text{s}}}$ a realization of $M_{n_{\text{s}}}$. If the following inequalities hold,
        \begin{align}
           & \tilde{m}_{n_{\text{s}}} > 1 - \epsilon \label{eq:theorem1_cond1} \\
           &  n_{\text{s}} > \ln \left ( \frac{1}{\beta} \right ) \frac{1}{  \left ( \tilde{m}_{n_{\text{s}}} + \epsilon \right ) \ln (\tilde{m}_{n_{\text{s}}} + \epsilon) - ( \tilde{m}_{n_{\text{s}}}+ \epsilon - 1)}, \label{eq:theorem1_cond2}
        \end{align}
        then $ [1-\epsilon,1] $ is a confidence interval for the success probability $\nu$ of $X$ with the confidence level over $1-\beta$.
\end{theorem}
The proof is given in Appendix~\ref{sec:thm1proof}. In our problem, $\tilde{m}_{n_{\text{s}}}$ is a realization of an estimator of the success probability $\nu_{u} (t+1)$ obtained from realizations of $X_u (t+1)$. This theorem implies that, if both $\tilde{m}_{n_{\text{s}}}$ and the number of samples $n_{\text{s}}$ are sufficiently large, then $\nu_{u} (t+1)$ is larger than $1-\epsilon$. Thus, to verify whether or not $\nu_u (t+1)$ is larger than $1-\epsilon$, the utility can obtain a number of realizations of $X_{u} (t+1)$ and check if inequalities \eqref{eq:theorem1_cond1} and \eqref{eq:theorem1_cond2} hold.
    
Now, we introduce the procedure the utility uses to obtain realizations of $X_u (t+1)$ given some $u \in [ -1, 1 ]$. The utility first computes the probability mass function (pmf) of $\bm{N}^{\text{ON}} (t)$ given the observed $\bm{p} (t)$ and $\bm{q} (t)$ as follows,
\begin{equation} \label{eq:nOnTCLsdist}
    \begin{aligned}
        & \mathrm{Pr}  \Big ( \bm{N}^{\text{ON}} (t) = \bm{n}^{\text{ON}} \; \big | \; ( \bm{P} (t)  = \bm{p} (t) )\cap ( \bm{Q} (t) = \bm{q} (t) ) \Big ) = \\
        & \mathrm{Pr} \Big ( \left( \bm{P}^{\text{L}} (t) = \bm{p} (t) - \Xi_{p} \bm{n}^{\text{ON}} \right )\cap  \left ( \bm{Q}^{\text{L}} (t) = \bm{q} (t) - \Xi_{q} \bm{n}^{\text{ON}} \right ) \\
        & \qquad \qquad \qquad \qquad \qquad \quad | \; \left( \bm{P} (t)  = \bm{p} (t) \right ) \cap \left ( \bm{Q} (t) = \bm{q} (t) \right ) \Big ) \\ 
        & \quad = \frac{ f_{\bm{P}^{\text{L}},\bm{Q}^{\text{L}}}  \left ( \bm{p} (t) - \Xi_{p} \bm{n}^{\text{ON}} , \bm{q} (t) - \Xi_{q} \bm{n}^{\text{ON}} \right )}{ \sum_{\bm{n} \in \mathbb{N}^{\text{ON}}} f_{\bm{P}^{\text{L}},\bm{Q}^{\text{L}}} \left ( \bm{p} (t) - \Xi_{p} \bm{n} , \bm{q} (t) - \Xi_{q} \bm{n} \right ) },
    \end{aligned}
\end{equation}
    where $\mathbb{N}^{\text{ON}} := \left \{ \bm{n}^{\text{ON}} \; | \; n_{j}^{\text{ON}} \in [ n_{j}^{\text{TCL}} ]_0 \enspace \forall j \in [n] \right \}$ is the set of all possible vectors for $\bm{N}^{\text{ON}} (t)$. Then, the utility obtains a realization $\tilde{x}_u (t+1)$ of $X_{u} (t+1)$ through the following sampling procedure, illustrated in Fig.~\ref{fig:test_procedure}.
    \begin{enumerate}[Step 1)]
    \item a. Obtain a realization $\tilde{\bm{n}}^{\text{ON}} (t)$ of $\bm{N}^{\text{ON}} (t)$ by sampling from its pmf derived through \eqref{eq:nOnTCLsdist}, and compute $\tilde{\bm{n}}^{\text{OFF}} (t) = \bm{n}^{\text{TCL}} - \tilde{\bm{n}}^{\text{ON}} (t)$.
    \, b. Obtain realizations $\tilde{\bm{p}}^{\text{L}} (t+1)$ and $\tilde{\bm{q}}^{\text{L}} (t+1)$ of $\bm{P}^{\text{L}} (t+1)$ and $\bm{Q}^{\text{L}} (t+1)$ by sampling  from $f_{\bm{P}^{\text{L}},\bm{Q}^{\text{L}}} ( t+1 )$.
        
        \item Obtain realizations $ \tilde{\bm{s}}^{\text{ON}}$ and $\tilde{\bm{s}}^{\text{OFF}}$ of $\bm{S}^{\text{ON}} (t+1)$ and $\bm{S}^{\text{OFF}} (t+1)$ by computing their elements per \eqref{eq:numofIntswitchTCLs} as
        \begin{equation*}
            \begin{aligned}
            & \tilde{s}_{j}^{\text{ON}} (t+1)  = \hat{w}_{j}^{\text{ON}} (t+1) \tilde{n}_{j}^{\text{OFF}} (t) \enspace \forall j \in [n] \\ & \tilde{s}_{j}^{\text{OFF}} ( t+1) =  \hat{w}_{j}^{\text{OFF}} (t+1) \tilde{n}_{j}^{\text{ON}} (t) \enspace \forall j \in [n].
            \end{aligned}
        \end{equation*}
        
        \item Obtain realizations $\tilde{\bm{c}}_u^{\text{ON}} (t+1)$ and $\tilde{\bm{c}}_u^{\text{OFF}} (t+1)$ of $\bm{C}_{u}^{\text{ON}} (t+1)$ and $\bm{C}_{u}^{\text{OFF}} (t+1)$ by sampling their elements per \eqref{eq:numofswitchTCLs} from the binomial distributions $\mathcal{B} \left ( \tilde{n}_{j}^{\text{OFF}} (t) - \tilde{s}_{j}^{\text{ON}} (t+1), u^{+} \right )$ and $\mathcal{B} \left ( \tilde{n}_{j}^{\text{ON}} (t) - \tilde{s}_{j}^{\text{OFF}} ( t+1 ), u^{-} \right )$.
        
        \item Obtain realizations of $\bm{N}_u^{\text{ON}} (t+1)$, $\bm{P}_{u} (t+1)$, $\bm{Q}_{u} (t+1)$, $\bm{V}_{u} (t+1)$, and $X_u (t+1)$ as
        \begin{equation*}
            \begin{aligned}
                 & \tilde{\bm{n}}_u^{\text{ON}} ( t+1 )  = \tilde{\bm{n}}^{\text{ON}} (t) + \tilde{\bm{s}}^{\text{ON}} (t+1) - \tilde{\bm{s}}^{\text{OFF}} ( t+1) \\
                 & \qquad \qquad \qquad \qquad \qquad + \tilde{\bm{c}}_u^{\text{ON}} ( t+1) - \tilde{\bm{c}}_u^{\text{OFF}} ( t+1) \\ 
                 & \tilde{\bm{p}}_u ( t+1) = \tilde{\bm{p}}^{\text{L}} ( t+1) + \Xi_{p} \tilde{\bm{n}}_u^{\text{ON}} (t+1) \ \\
                 & \tilde{\bm{q}}_u (t+1) = \tilde{\bm{q}}^{\text{L}} ( t+1)  + \Xi_{q} \tilde{\bm{n}}_u^{\text{ON}} ( t+1) \\
                 & \tilde{v}_{u,j} ( t+1) = f_{v_j} ( \tilde{\bm{p}}_u ( t+1) ,  \tilde{\bm{q}}_u ( t+1), v_0) \enspace \forall j \in [n] \\
                 & \tilde{x}_u ( t+1) = \mathbbm{1} ( \text{min}_{j \in [n]} \tilde{v}_{u,j} ( t+1) \geq \underline{v} ). 
            \end{aligned}
        \end{equation*}
    \end{enumerate}
    
    \begin{figure}[!t]
        \centering
        \includegraphics[width=0.47\textwidth]{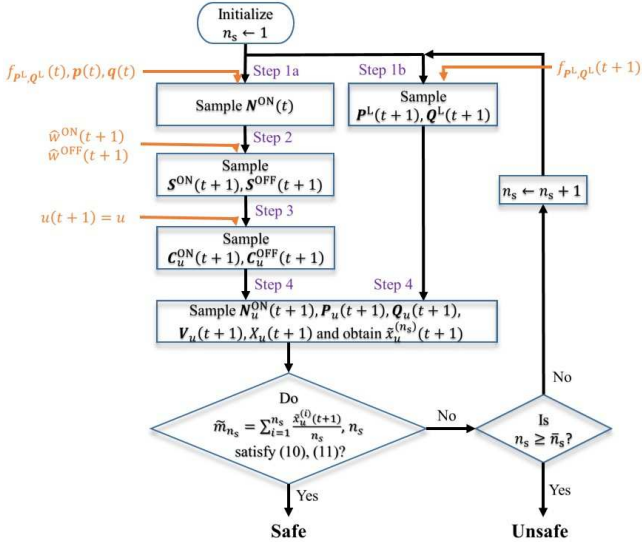}
        \vspace{-.5cm}
        \caption{Flowchart of the test procedure to check if a one-step ahead command $u (t+1) = u$ satisfies the chance constraint. The information required for each step is in orange.}\vspace{-0.5cm}
        \label{fig:test_procedure}
    \end{figure}

The utility can 
obtain multiple realizations of $X_{u} (t+1)$ by iterating this sampling procedure. Denote each realization $i$ of $X_u (t+1)$ as $\tilde{x}_u^{(i)} (t+1)$, where $i \in [ n_{\text{s}} ]$. In each iteration, the utility updates the realization of the estimator $\tilde{m}_{n_{\text{s}}} = \sum_{i=1}^{n_{\text{s}}} \tilde{x}_u^{(i)} (t+1) / n_{\text{s}}$ and checks if the inequalities \eqref{eq:theorem1_cond1}, \eqref{eq:theorem1_cond2} hold. If they do, $u (t+1) = u$ satisfies the chance constraint with confidence level over $1 - \beta$; otherwise, the utility continues to iterate until $n_{\text{s}}$ reaches some pre-determined upper bound $\overline{n}_{\text{s}}$, as shown in Fig.~\ref{fig:test_procedure}. 
    
    Next, we construct a one-step ahead constraint set $\mathcal{U} (t+1)$. We first make an assumption on the monotonicity of $\nu_{u} ( t+1 )$.
    \begin{assumption} \label{ass:prob_monotonicity}
        The one-step ahead probability of network safety $ \nu_{u} (t+1)  $  monotonically decreases with respect to $u$.
    \end{assumption}
    
The intuition behind this assumption is that the real and reactive power consumption at each node is likely to increase as $u$ increases, which is also likely to lead to a voltage decrease at every node. This assumption will be justified in Section~\ref{subsec:monotonicity_justification}. Under this assumption, the following holds.
    \begin{theorem} \label{thm:solution}
        Suppose that Assumption~\ref{ass:prob_monotonicity} holds and let $\tilde{x}_{\overline{u}}^{(1)} ( t+1) , \ldots , \tilde{x}_{\overline{u}}^{(n_{\text{s}})} ( t+1) $ be $n_{\text{s}}$ realizations of $X_{\overline{u}} (t+1)$ for a command $\overline{u} \in [-1,1]$. If $n_{\text{s}}$ and $\tilde{m}_{n_{\text{s}}} = \sum_{i=1}^{n_{\text{s}}} \tilde{x}_{\overline{u}}^{(i)} ( t+1 )  / n_{\text{s}}$ satisfy \eqref{eq:theorem1_cond1} and \eqref{eq:theorem1_cond2}, then $\mathcal{U} (t+1) = [-1 , \overline{u} ]$ is a solution to the Problem~\ref{prob:problem_statement} with confidence level over $1-\beta$.
    \end{theorem}
    \begin{proof}  
    By Theorem~\ref{thm:confidence_interval_condition}, the interval $[1-\epsilon,1]$ is a confidence interval for $\nu_{\overline{u}} ( t+1 )$ with confidence level over $1-\beta$. Also, under Assumption~\ref{ass:prob_monotonicity}, $ \nu_{u} (t+1) \geq \nu_{\overline{u}} (t+1)$ holds for any $u \in \mathcal{U} (t+1) = [-1,\overline{u}]$. Thus, $\nu_{u} (t+1) $ is greater than or equal to $1 - \epsilon$ for any $u \in \mathcal{U} (t+1)$ with confidence level over $1 - \beta$.
    \end{proof}
    
    This theorem means that, if the one-step ahead probability of network safety $\nu_{\overline{u}} ( t+1 )$ under the command $ u (t+1) = \overline{u}$ is greater than or equal to the desired safety probability, then any less aggressive command in the range $[-1, \overline{u}]$ also satisfies the chance constraint. Therefore, a solution to Problem~\ref{prob:problem_statement} is the interval $[-1, \overline{u}]$, where $\overline{u}$ passes the test procedure in Fig.~\ref{fig:test_procedure}. 
        
    The choice of probabilistically-safe set $\mathcal{U} (t+1)$ is not unique. A larger $\mathcal{U} (t+1)$ gives more flexibility to the aggregator, potentially improves the quality of balancing services, and reduces the conservativeness of our approach. Therefore, the utility should find the largest possible $\overline{u}$ that passes the test procedure. This can be achieved using the bisection method \cite{burden2015numerical}, starting with $\overline{u}=1$.

    \begin{remark} \label{remark:over_voltage}
        To restrict the probability of over-voltage violations, we can also apply the monotonicity assumption; the probability of over-voltage violations increases as the command $u$ decreases. In this case, we can use the bisection method to obtain a lower bound on $u (t+1)$. Then, the utility can send both a lower and upper bound on $u (t+1)$ to restrict the probability of over- and under-voltage violations.
    \end{remark}
    
     \begin{remark} \label{remark:errorfromest}
     Since the utility approximates $\bm{P}^{\mathrm{T}}$ and $\bm{Q}^{\mathrm{T}}$ in \eqref{eq:TCLpowerapprox} and uses estimates of $\bm{w}^{\mathrm{ON}}$ and $\bm{w}^{\mathrm{OFF}}$ in Step 2 of the sampling procedure, Theorem~\ref{thm:solution} holds only if those approximations/estimates are accurate. We will justify the use of these approximations/estimations through simulation in Section~\ref{sec:numerical_experiment}.
    \end{remark}

    \subsection{Justification of Assumption~\ref{ass:prob_monotonicity}} \label{subsec:monotonicity_justification}

    In this section, we justify Assumption~\ref{ass:prob_monotonicity} by showing that an approximation of $\nu_{u} (t+1)$ is a monotonically decreasing function with respect to $u$. 
    We consider the LinDistFlow equations~\cite{baran1989network}, which drop the nonlinear terms of~\eqref{eq:branch_flow_equations}, i.e.,
    \begin{equation}
    \label{eq:approx_branch_flow_equation}
    \begin{aligned} 
            \hat{p}_{j}^{\text{b}} = & \sum_{k \in c(j)} \hat{p}_{k}^{\text{b}} + p_j, \quad \hat{q}_{j}^{\text{b}} = \sum_{k \in c(j)} \hat{q}_{k}^{\text{b}} + q_j  \\
            \hat{v}_{j}^2 = & \; \hat{v}_{e(j)}^2 - 2 (r_j \hat{p}_{j}^{\text{b}} +  x_j \hat{q}_{j}^{\text{b}}), 
    \end{aligned}
    \end{equation}
    where variables with hats correspond to approximations of the original DistFlow variables.
    Let $ \hat{f}_{v_{j}} (\bm{p}, \bm{q} , v_0)$ be the voltage solution of \eqref{eq:approx_branch_flow_equation}, i.e., $ \hat{V}_{u,j} (t+1) := \hat{f}_{v_{j}} (\bm{P}_u (t+1), \bm{Q}_u (t+1), v_0)$ is the approximate voltage at node $j$. Also, let 
     \begin{equation} 
        \hat{X}_{u} (t+1) = \mathbbm{1} \left ( \text{min}_{j \in [n]} \hat{V}_{u,j} ( t+1 )  \geq \underline{v} \right ),
    \end{equation}
    whose success probability  $\hat{\nu}_{u} (t+1) := \mathrm{Pr} ( \hat{X}_{u} (t+1) =1 ) $ approximates $\nu_{u} (t+1)$. To show that $\hat{\nu}_{u} (t+1)$ is decreasing with respect to $u$, we start with a proposition.
    \begin{proposition} \label{ass:voltage_monotonicity}
        Suppose that $\bm{p}^{(1)},\bm{p}^{(2)} \in \mathbb{R}^n$ and $\bm{q}^{(1)},\bm{q}^{(2)} \in \mathbb{R}^n$ are different instances of real and reactive power consumption where $p_{j}^{(1)} \leq p_{j}^{(2)} $ and $q_{j}^{(1)} \leq q_{j}^{(2)} \, \forall \, j \in [n]$. Then, $\hat{f}_{v_j} (\bm{p}^{(1)},\bm{q}^{(1)}, v_0) \geq \hat{f}_{v_j} (\bm{p}^{(2)} , \bm{q}^{(2)} , v_0)$ for all $j \in [n]$.  
    \end{proposition}
    
    \begin{proof}   
    First let $\hat{f}_{p_{j}^{\text{b}}} (\bm{p},\bm{q},v_0)$ and $\hat{f}_{q_{j}^{\text{b}}} (\bm{p},\bm{q},v_0)$ be the solutions of \eqref{eq:approx_branch_flow_equation} corresponding to $p_{j}^{\text{b}}$ and $q_{j}^{\text{b}}$ when the real and reactive power consumption at each node are $\bm{p}$ and $\bm{q}$, and the substation voltage is $v_0$. Then, for all $j \in [n]$ \cite{baran1989network}
    \begin{equation} \label{eq:prop1_eq1}
    \begin{aligned}
       & \hat{f}_{p_{j}^{\text{b}}} ( \bm{p}, \bm{q}, v_0) = \sum_{k \in d(j) } p_{k}, \quad \hat{f}_{q_{j}^{\text{b}}} (\bm{p},\bm{q}, v_0) = \sum_{k \in d(j) } q_{k} \\
        & \hat{f}_{v_j}^2   ( \bm{p}, \bm{q}, v_0 ) = v_0^2   - 2 \sum_{k \in a(j) } \Big ( r_k \hat{f}_{p_{k}^{\text{b}}} ( \bm{p}, \bm{q}, v_0 ) \\
        & \qquad \qquad \qquad  \qquad \qquad \qquad \quad + x_k \hat{f}_{q_{k}^{\text{b}}} ( \bm{p}, \bm{q}, v_0 ) \Big ),
    \end{aligned}
    \end{equation}
    where $d(j) := c(j) \cup \{ j \}$ is the set of indices of all descendants of node $j$ including itself, and $a(j)$ is the set of indices of all ancestors of node $j$ including itself. Hence, $\hat{f}_{p_{j}^{\text{b}}} (\bm{p} ,\bm{q},v_0)$ and $ \hat{f}_{q_{j}^{\text{b}}} ( \bm{p} ,\bm{q}, v_0)$ are increasing as $p_k$ and $q_k$ increase for all $k \in [n]$, and $ \hat{f}_{p_{j}^{\text{b}}} (\bm{p}^{(1)},\bm{q}^{(1)}, v_0) \leq \hat{f}_{p_{j}^{\text{b}}} (\bm{p}^{(2)}, \bm{q}^{(2)}, v_0)$ and $ \hat{f}_{q_{j}^{\text{b}}} (\bm{p}^{(1)},\bm{q}^{(1)}, v_0) \leq \hat{f}_{q_{j}^{\text{b}}} (\bm{p}^{(2)}, \bm{q}^{(2)}, v_0)$  for all $j \in [n]$. Also, since all $r_k$ and $x_k$ are positive,  $\hat{f}_{v_j}  ( \bm{p}, \bm{q}, v_0 )$ is decreasing as $ \hat{f}_{p_{k}^{\text{b}}} ( \bm{p}, \bm{q}, v_0)$ and $ \hat{f}_{q_{k}^{\text{b}}} (\bm{p},\bm{q}, v_0)$ increase for all $ k \in [n]$. Therefore, $\hat{f}_{v_j} (\bm{p}^{(1)},\bm{q}^{(1)}, v_0) \geq \hat{f}_{v_j} (\bm{p}^{(2)} , \bm{q}^{(2)}, v_0) \, \forall \, j \in [n]$.
    \end{proof}

    This proposition states that $\hat{f}_{v_j} ( \bm{p}, \bm{q}, v_0 )$ monotonically decreases as the real and reactive power consumption $p_j$ and $q_j$ at every node increase for all $j \in [n]$. Since the one-step ahead real and reactive power consumption of the TCLs at each node are likely to increase as $u$ increases (recall that in Section~\ref{sec:aggregator_TCLs} we made the realistic assumption that TCLs have constant lagging power factors, and so their real and reactive power consumption change in the same direction), this proposition implies that the probability of under-voltage violations increases as $u$ increases. This is stated in the following theorem.
    \begin{theorem} \label{thm:prob_monotonicity}
        The approximate probability of network safety $\hat{\nu}_{u} (t+1)$ under the one-step ahead command $ u$ is a monotonically decreasing function of $u$. 
    \end{theorem}
    The proof is given in Appendix~\ref{sec:proof_monotonicity}. While Theorem~\ref{thm:prob_monotonicity} justifies Assumption~\ref{ass:prob_monotonicity} for the approximation $\hat{\nu}_{u} (t+1)$, we also empirically validate that $\nu_u (t+1)$ is a monotonically decreasing function of $u$ in Fig.~\ref{fig:prob_empirical_estimation}. To create this plot, we generated $n_{\text{s}} = 10^6$ realizations of $X_{u} ( t+1 )$ for each of 101 uniformly spaced points $u$ from -1 to 1. 

\begin{figure}[t]
    \centering
    \includegraphics[width=0.7\linewidth]{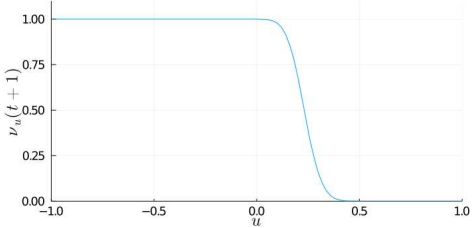}
    \vspace{-.3cm}
    \caption{Demonstration of the monotonicity of $\nu_{u} (t+1)$ with respect to $u$.}
    \label{fig:prob_empirical_estimation}
    \vspace{-.3cm}
\end{figure}
    
    
\section{Case Study} \label{sec:numerical_experiment}

We next present the result of a case study in which we compare the proposed approach with two benchmark approaches, a tracking controller benchmark and an Optimal Power Flow (OPF) benchmark. We first describe our simulation setup and detail the benchmark approaches. Then, we present our results. 

We use the 56-bus balanced distribution network from \cite{bolognani2015existence} where the nominal real and reactive power consumption at node $j$ are denoted by $p_{j}^{\mathrm{Ln}}$ and $q_{j}^{\mathrm{Ln}}$, respectively. We set the safe lower bound on the voltage to $\underline{v} = 0.95 \; \text{pu}$. TCL parameters are randomly sampled\footnote{Each parameter is sampled from uniform distributions with intervals: $\theta_a^i \in [ 29, 31 ]$ \textdegree{}C, $ c_{\text{th}}^{i} \in [ 1.5, 2.5 ] $kWh/\textdegree{}C,  $r_{\text{th}}^i = [ 1.2, 2.5 ]$\textdegree{}C/kW, $p_{\text{tr}}^{i} \in [ -18, -14 ]$ kW, $\zeta^i \in [ 2.3, 2.7 ], $ $\theta_s^{i} \in [ 20, 25 ]$\textdegree{}C,  $ \overline{\theta}^i - \underline{\theta}^i \in [1.5,2]$\textdegree{}C, and $\omega^i = \tan (\arccos(\phi^i))$, where $\phi^i \in [0.95,0.99]$.}, and the TCLs are distributed throughout the network so that the aggregate TCLs' nominal real power consumption at node $j$ is approximately $0.25 p_{j}^{\text{Ln}}$.
  For simplicity, we assume that the real and reactive power consumption of the non-participating loads at each node $P_{j}^{\text{L}} (t)$ and $Q_{j}^{\text{L}} (t)$ follow normal distributions $ \mathcal{N} ( \overline{p}_{j}^{\text{Ln}} (t), ( 0.15 p_{j}^{\text{Ln}})^2 )$ and $\mathcal{N} (  \overline{q}_{j}^{\text{Ln}} (t) , ( 0.15 q_{j}^{\text{Ln}} )^2 )$ truncated by the intervals $[ p_{j}^{\text{Lmin}}, p_{j}^{\text{Lmax}}] = [ -0.25 p_{j}^{\text{Ln}}, 0.675 p_{j}^{\text{Ln}}]$ and $ [q_{j}^{\text{Lmin}}, q_{j}^{\text{Lmax}}] = [ -0.25 q_{j}^{\text{Ln}}, 0.675 q_{j}^{\text{Ln}}]$, respectively. We conduct 2h simulations (13h-15h) and let $ \overline{p}_{j}^{\text{Ln}} (t)$ and $\overline{q}_{j}^{\text{Ln}} (t)$  linearly increase from $0.5$ to $0.65$ of their nominal values from 13.0h to 13.9h, stay constant from 13.9h to 14.1h, and linearly decrease to $0.5$ of their nominal values from 14.1h to 15.0h. The reference signal $p_{\text{ref}} (t)$ is a scaled and shifted 2h segment of a PJM RegD signal \cite{PJMrefsignal}. We use the desired safety probabilities $1 - \epsilon = 0.95$ and $0.98$ and the desired confidence level $ 1 - \beta = 0.999$. 

The aggregator obtains the estimates $\hat{w}_{j}^{\text{ON}} (t+1)$ and $\hat{w}_{j}^{\text{OFF}} ( t+1 )$ for each node leveraging an approximate model of the dynamics of the TCL aggregation. The model was developed in past work, e.g.,~\cite{mathieu2012state}, and so not detailed here. While we could identify different models for each node, here we use the same model for each node $j \in [n]$ and so $\hat{w}_{j}^{\text{ON}} (t+1)$ and $\hat{w}_{j}^{\text{OFF}} ( t+1 )$ are identical across nodes. Fig.~\ref{fig:Intswitchport} demonstrates the model's estimation performance, showing the actual and estimated fractions of TCLs outside of their dead-bands. Although the estimates do not perfectly track the actual values, they capture the overall trends.

\begin{figure}
    \centering
        \includegraphics[width=.48\columnwidth]{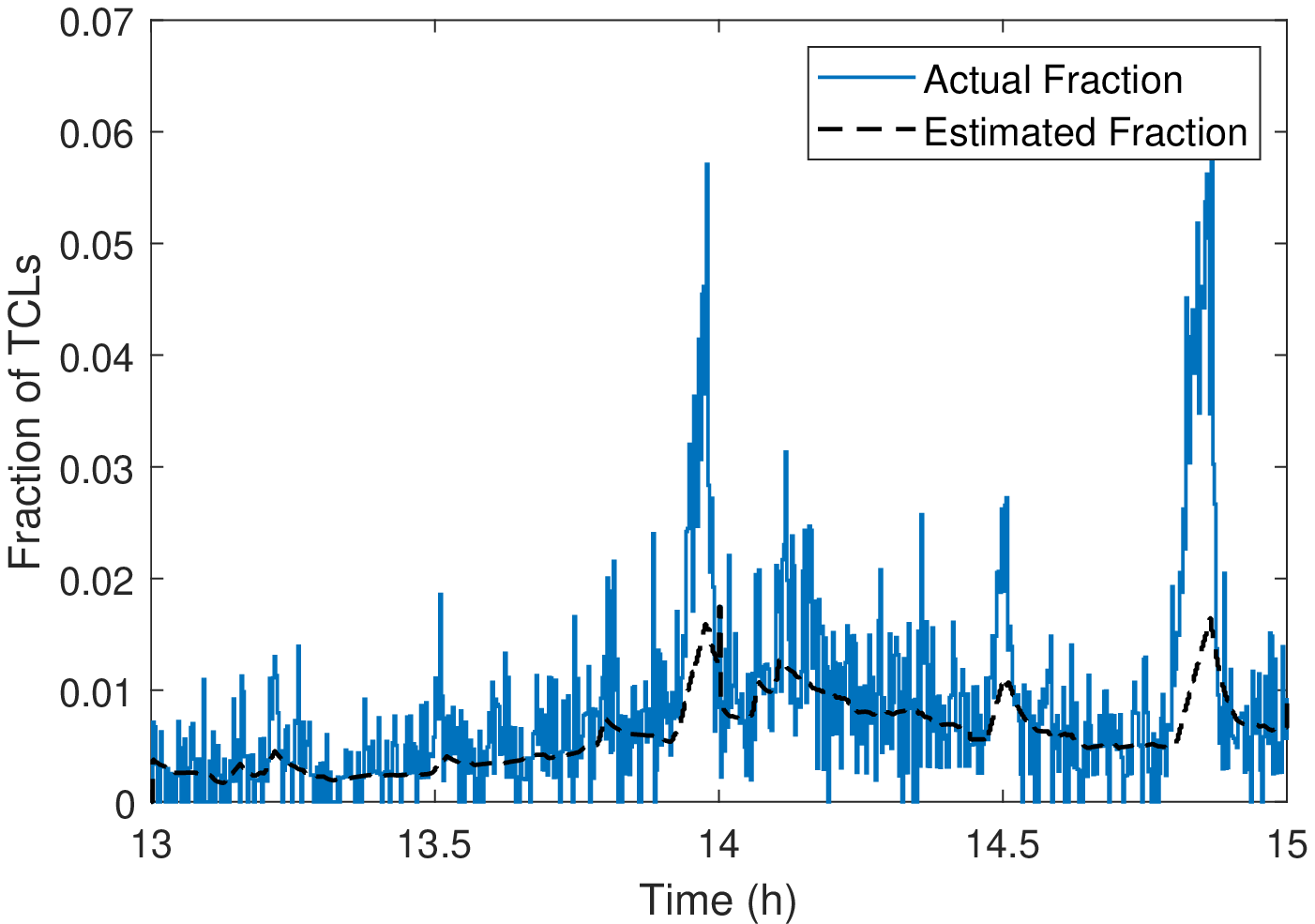}
        \includegraphics[width=.48\columnwidth]{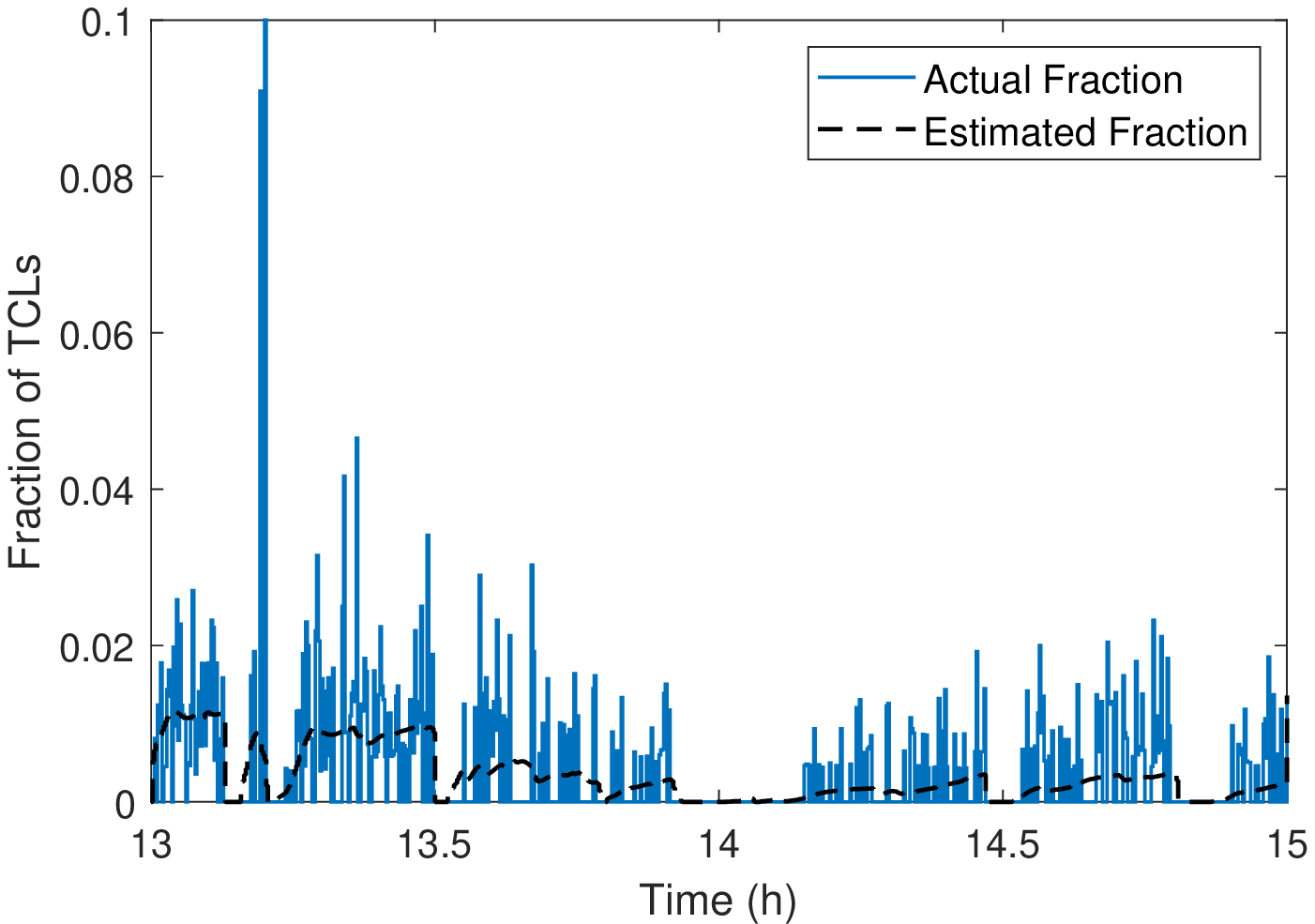}
        \vspace{-.3cm}
    \caption{Actual and estimated fractions of TCLs switched ON (left) and OFF (right) by their thermostats, in the proposed approach $(\epsilon =0.02)$.}
    \label{fig:Intswitchport}
    \vspace{-.5cm}
\end{figure}

The tracking controller benchmark does not take into account network safety. It chooses the optimal command $u_{\text{opt}} ( t )$ using \eqref{eqn:uopt} with $\mathcal U(t) = [-1, 1]$,
where $\mathbb{E} \left [ P_{\text{agg}} ( t ) \right ]$ is the expected aggregate power of the TCLs under $u (t) = u$, which is computed with the same approximate aggregate TCL model. 

The OPF benchmark approximately enforces network safety assuming linearized power flow. It solves the following mixed integer linear program at each time step to compute the optimal one-step ahead mode of each TCL,
\begin{subequations}
\label{eq:benchmark}
\begin{align}
    \min_{m^i} \; & | p_{\text{agg}} - p_{\text{ref}}  | \\
    \text{s.t.} \; & p_{\text{agg}} =  \sum_{i=1}^{n^{\text{TCL}}} p^{i} m^{i} \\
    & p_{j}^{\text{T}} =  \sum_{i \in \mathcal{I}_j} p^{i} m^{i}, \enspace  q_{j}^{\text{T}} =  \sum_{i \in \mathcal{I}_j} q^{i} m^{i}, \quad\quad \forall j \in [n]  \\
    & \text{TCL temperature dynamics } \eqref{eq:TCLtempdyns}, \quad \, \, \forall i \in [n^{\text{TCL}}] \\
    &  \theta^i \in [\underline{\theta^i},\overline{\theta^i}], \qquad \qquad \qquad \qquad \, \, \quad  \forall i \in [n^{\text{TCL}}] \\
     & \bm{v} = \bm{\Phi}_{p} (\bm{p}^{\text{T}}+\bm{p}^{\text{Lmax}}) + \bm{\Phi}_{q} (\bm{q}^{\text{T}}+\bm{q}^{\text{Lmax}}) + \bm{\Phi}_{c} \label{eq:linPF} \\
    & \underline{v} \leq \bm{v} \label{eq:OPFdetconst},
\end{align}
\end{subequations}
where $\mathcal{I}_j$ is the set of indices of TCLs connected to $j$ and \eqref{eq:linPF} is the linearized power flow developed in \cite{dall2017optimal}. The OPF benchmark is different from the proposed approach and optimal tracking controller in that it can observe the TCLs' internal temperatures and directly control the TCLs' modes. In contrast to the proposed approach, it has a deterministic constraint \eqref{eq:OPFdetconst} on network safety rather than a chance constraint.

\begin{figure*}
    \centering
    \includegraphics[width=\textwidth]{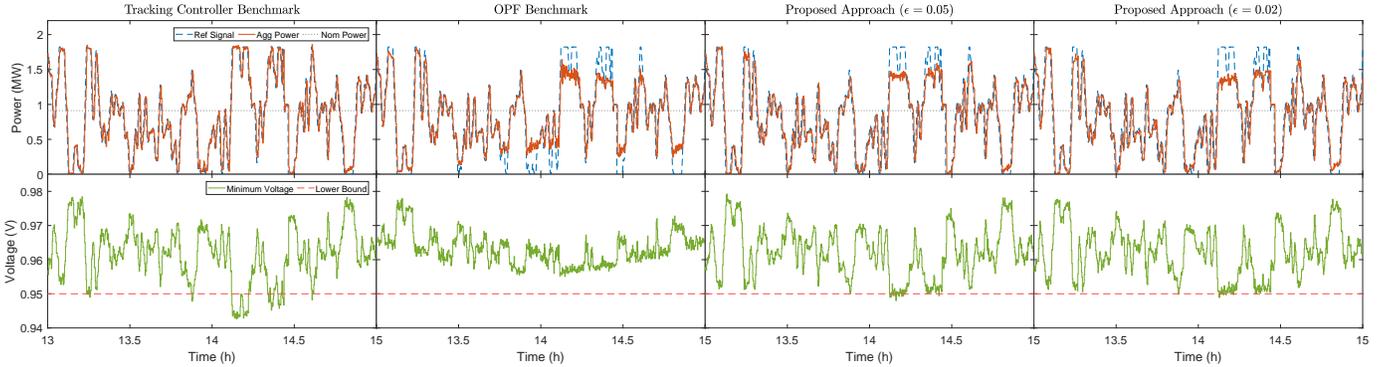}
      \vspace{-.8cm}
    \caption{The reference signal and the TCLs' aggregate power (top), and the minimum network voltage and the safe lower bound (bottom) for each algorithm.}
    \label{fig:ExpResult}
    \vspace{-.5cm}
\end{figure*}

Fig.~\ref{fig:ExpResult} illustrates the results of the comparison between the two benchmarks and our proposed approach with $\epsilon = 0.05$ and $0.02$. Table~\ref{tab:ExpResult} shows the root mean squared error (RMSE) of the aggregate power from the reference signal, along with the empirical safety probability computed as the fraction of time steps in which under-voltage violations (computed with the nonlinear power flow equations) do not happen. The tracking controller benchmark has the best tracking performance, but frequently causes under-voltage violations. This demonstrates the need to employ network-safe DER control strategies. In contrast, the OPF benchmark avoids under-voltage violations, but has the worst tracking performance, demonstrating that approaches that (approximately) enforce network safety will at times have poor balancing service performance. 

\begin{table}[t]
    \centering
    \caption{Tracking and Safety Performance of Each Algorithm}
    \vspace{-.1cm}
    \begin{tabular}{lcccc}
    \hline
    & Track Ctrl & OPF & \multicolumn{2}{c}{Proposed Approach}  \\
    & Benchmark & Benchmark & $\epsilon = 0.05$ & $\epsilon = 0.02$  \\
    \hline \hline
     RMSE (kW) & $ 77.05 $ 
      & $ 168.3 $ & $102.8$ & $118.8$ \\
      Safety Probability \hspace{-.2cm} & 0.908 & 1.00 & 0.981 & 0.986 \\
      \hline
    \end{tabular}
        \vspace{-.5cm}
        \label{tab:ExpResult}
\end{table} 

Our approach achieves a better trade-off between tracking performance and network safety; specifically, it achieves better tracking performance than the OPF benchmark and satisfies the chance constraint on network safety, resulting in fewer under-voltage violations than the tracking controller benchmark. As shown in Table~\ref{tab:ExpResult}, the empirical safety probabilities are over the target values $1-\epsilon$. The RMSE increases as $\epsilon$ decreases, which is expected since higher $1-\epsilon$ results in more conservative bounds on the input commands.

\section{Conclusion}
This paper proposed an approach to coordinate a collection of TCLs to provide balancing services while guaranteeing network safety with high probability. In particular, we proposed a constraint construction method that would allow the utility to constrain the input commands of an aggregator providing balancing services like frequency regulation. The approach imposes a chance constraint on network safety, wherein both the violation probability and confidence level are design parameters that can be selected by the utility. We used the bisection method to compute the largest possible constraint set, which provides the most flexibility to the aggregator. 
    
Future work will extend the proposed approach to incorporate different types of DERs, such as stationary batteries, electric vehicles, and curtailable solar photovoltaics, into the framework; we already have some preliminary work along this direction~\cite{jang2022data}.

\bibliographystyle{IEEEtran}
\bibliography{reference} 

\begin{thebibliography}{10}
\providecommand{\url}[1]{#1}
\csname url@rmstyle\endcsname
\providecommand{\newblock}{\relax}
\providecommand{\bibinfo}[2]{#2}
\providecommand\BIBentrySTDinterwordspacing{\spaceskip=0pt\relax}
\providecommand\BIBentryALTinterwordstretchfactor{4}
\providecommand\BIBentryALTinterwordspacing{\spaceskip=\fontdimen2\font plus
\BIBentryALTinterwordstretchfactor\fontdimen3\font minus
  \fontdimen4\font\relax}
\providecommand\BIBforeignlanguage[2]{{%
\expandafter\ifx\csname l@#1\endcsname\relax
\typeout{** WARNING: IEEEtran.bst: No hyphenation pattern has been}%
\typeout{** loaded for the language `#1'. Using the pattern for}%
\typeout{** the default language instead.}%
\else
\language=\csname l@#1\endcsname
\fi
#2}}

\bibitem{dall2017optimal}
E.~Dall’Anese, S.~Guggilam, A.~Simonetto, Y.~C. Chen, and S.~Dhople,
  ``Optimal regulation of virtual power plants,'' \emph{IEEE Trans. Power
  Syst.}, vol.~33, no.~2, pp. 1868--1881, 2017.

\bibitem{bernstein2019real}
A.~Bernstein and E.~Dall’Anese, ``Real-time feedback-based optimization of
  distribution grids: A unified approach,'' \emph{IEEE Trans. Control Netw.
  Syst.}, vol.~6, no.~3, pp. 1197--1209, 2019.

\bibitem{vrettos2013combined}
E.~Vrettos and G.~Andersson, ``Combined load frequency control and active
  distribution network management with thermostatically controlled loads,'' in
  \emph{IEEE Int. Conf. Smart Grid Comm.}, 2013, pp. 247--252.

\bibitem{ferc2018notc}
{Federal Energy Regulatory Commission}, ``Notice inviting post-technical
  conference comments,'' Tech. Rep. Docket No.. RM18-9-000, Apr. 2018.
  [Online]. Available :
  \url{https://cms.ferc.gov/sites/default/files/2020-09/Notice-for-Comments-on-NOPR-RM18-9.pdf}.

\bibitem{ferc2222}
------, ``{FERC} order no. 2222: Participation of distributed energy resource
  aggregations in markets operated by regional transmission organizations and
  independent system operators.''
  \url{https://www.ferc.gov/sites/default/files/2020-09/E-1_0.pdf}, Sep 2020.

\bibitem{mathieu2012state}
J.~Mathieu, S.~Koch, and D.~S. Callaway, ``State estimation and control of
  electric loads to manage real-time energy imbalance,'' \emph{IEEE Trans.
  Power Syst.}, vol.~28, no.~1, pp. 430--440, 2012.

\bibitem{bashash2012modeling}
S.~Bashash and H.~K. Fathy, ``Modeling and control of aggregate air
  conditioning loads for robust renewable power management,'' \emph{IEEE Trans.
  Control Syst. Technol.}, vol.~21, no.~4, pp. 1318--1327, 2012.

\bibitem{zhang2013aggregated}
W.~Zhang, J.~Lian, C.-Y. Chang, and K.~Kalsi, ``Aggregated modeling and control
  of air conditioning loads for demand response,'' \emph{IEEE Trans. Power
  Syst.}, vol.~28, no.~4, pp. 4655--4664, 2013.

\bibitem{tindemans2015decentralized}
S.~Tindemans, V.~Trovato, and G.~Strbac, ``Decentralized control of
  thermostatic loads for flexible demand response,'' \emph{IEEE Trans. Control
  Syst. Technol.}, vol.~23, no.~5, pp. 1685--1700, 2015.

\bibitem{ross2019coordination}
S.~C. Ross, N.~Ozay, and J.~L. Mathieu, ``Coordination between an aggregator
  and distribution operator to achieve network-aware load control,'' in
  \emph{PowerTech}, 2019.

\bibitem{ross2021strategies}
S.~Ross and J.~Mathieu, ``Strategies for network-safe load control with a
  third-party aggregator and a distribution operator,'' \emph{IEEE Trans. Power
  Syst.}, vol.~36, no.~4, pp. 3329--3339, 2021.

\bibitem{lee2021robust}
D.~Lee, K.~Turitsyn, D.~K. Molzahn, and L.~A. Roald, ``Robust {AC} optimal
  power flow with robust convex restriction,'' \emph{IEEE Trans. Power Syst.},
  vol.~36, no.~6, pp. 4953--4966, 2021.

\bibitem{nguyen2018constructing}
H.~D. Nguyen, K.~Dvijotham, and K.~Turitsyn, ``Constructing convex inner
  approximations of steady-state security regions,'' \emph{IEEE Trans. on Power
  Syst.}, vol.~34, no.~1, pp. 257--267, 2018.

\bibitem{nazir2021grid}
N.~Nazir and M.~Almassalkhi, ``Grid-aware aggregation and realtime
  disaggregation of distributed energy resources in radial networks,''
  \emph{IEEE Trans. on Power Syst.}, vol.~37, no.~3, pp. 1706--1717, 2021.

\bibitem{petrou2020operating}
K.~Petrou, M.~Z. Liu, A.~T. Procopiou, L.~F. Ochoa, J.~Theunissen, and
  J.~Harding, ``Operating envelopes for prosumers in {LV} networks: A weighted
  proportional fairness approach,'' in \emph{ISGT Europe}, 2020.

\bibitem{yi2022fair}
Y.~Yi and G.~Verbi{\v{c}}, ``Fair operating envelopes under uncertainty using
  chance constrained optimal power flow,'' \emph{Electr. Power Syst. Res.},
  vol. 213, p. 108465, 2022.

\bibitem{russell2022stochastic}
J.~S. Russell, P.~Scott, and A.~Attarha, ``Stochastic shaping of aggregator
  energy and reserve bids to ensure network security,'' \emph{Electr. Power
  Syst. Res.}, vol. 212, p. 108418, 2022.

\bibitem{comden2022secure}
J.~Comden, A.~S. Zamzam, and A.~Bernstein, ``Secure control regions for
  distributed stochastic systems with application to distributed energy
  resource dispatch,'' in \emph{ACC}, 2022, pp. 2208--2213.

\bibitem{ross2020method}
S.~C. Ross and J.~L. Mathieu, ``A method for ensuring a load aggregator’s
  power deviations are safe for distribution networks,'' \emph{Electr. Power
  Syst. Res.}, vol. 189, p. 106781, 2020.

\bibitem{jang2021large}
S.~Jang, N.~Ozay, and J.~L. Mathieu, ``Large-scale invariant sets for safe
  coordination of thermostatic loads,'' in \emph{ACC}, 2021, pp. 4163--4170.

\bibitem{baker2016distribution}
K.~Baker, E.~Dall'Anese, and T.~Summers, ``Distribution-agnostic stochastic
  optimal power flow for distribution grids,'' in \emph{NAPS}, 2016.

\bibitem{dall2017chance}
E.~Dall’Anese, K.~Baker, and T.~Summers, ``Chance-constrained ac optimal
  power flow for distribution systems with renewables,'' \emph{IEEE Trans.
  Power Syst.}, vol.~32, no.~5, pp. 3427--3438, 2017.

\bibitem{hassan2018chance}
A.~Hassan, Y.~Dvorkin, D.~Deka, and M.~Chertkov, ``Chance-constrained {ADMM}
  approach for decentralized control of distributed energy resources,'' in
  \emph{PSCC}, 2018.

\bibitem{ayyagari2017chance}
K.~S. Ayyagari, N.~Gatsis, and A.~F. Taha, ``Chance constrained optimization of
  distributed energy resources via affine policies,'' in \emph{GlobalSIP},
  2017, pp. 1050--1054.

\bibitem{hassan2018optimal}
A.~Hassan, R.~Mieth, M.~Chertkov, D.~Deka, and Y.~Dvorkin, ``Optimal load
  ensemble control in chance-constrained optimal power flow,'' \emph{IEEE
  Trans. Smart Grid}, vol.~10, no.~5, pp. 5186--5195, 2018.

\bibitem{chen2021combining}
Y.~Chen and Y.~Lin, ``Combining model-based and model-free methods for
  stochastic control of distributed energy resources,'' \emph{Applied Energy},
  vol. 283, p. 116204, 2021.

\bibitem{hart1992nonintrusive}
G.~W. Hart, ``Nonintrusive appliance load monitoring,'' \emph{Proc. IEEE},
  vol.~80, no.~12, pp. 1870--1891, 1992.

\bibitem{sonderegger1978dynamic}
R.~C. Sonderegger, \emph{Dynamic models of house heating based on equivalent
  thermal parameters.}\hskip 1em plus 0.5em minus 0.4em\relax Ph.D.
  dissertation, Princeton University, 1978.

\bibitem{baran1989optimal}
M.~Baran and F.~F. Wu, ``Optimal sizing of capacitors placed on a radial
  distribution system,'' \emph{IEEE Trans. Power Deliv.}, vol.~4, no.~1, pp.
  735--743, 1989.

\bibitem{kersting2018distribution}
W.~H. Kersting, ``Distribution system modeling and analysis,'' in
  \emph{Electric Power Generation, Transmission, and Distribution: The Electric
  Power Engineering Handbook}.\hskip 1em plus 0.5em minus 0.4em\relax CRC
  press, 2018, pp. 26--1.

\bibitem{burden2015numerical}
R.~L. Burden, J.~D. Faires, and A.~M. Burden, \emph{Numerical analysis}.\hskip
  1em plus 0.5em minus 0.4em\relax Cengage Learning, 2015.

\bibitem{baran1989network}
M.~E. Baran and F.~F. Wu, ``Network reconfiguration in distribution systems for
  loss reduction and load balancing,'' \emph{IEEE Power Eng. Rev.}, vol.~9,
  no.~4, pp. 101--102, 1989.

\bibitem{bolognani2015existence}
S.~Bolognani and S.~Zampieri, ``On the existence and linear approximation of
  the power flow solution in power distribution networks,'' \emph{IEEE Trans.
  Power Syst.}, vol.~31, no.~1, pp. 163--172, 2015.

\bibitem{PJMrefsignal}
PJM, ``{RTO} regulation signal data for 7.19.2019 \& 7.20.2019.xls,''
  \url{https://www.pjm.com/markets-and-operations/ancillary-services.aspx},
  accessed: 2019-10-22.

\bibitem{jang2022data}
S.~Jang, N.~Ozay, and J.~L. Mathieu, ``Data-driven estimation of probabilistic
  constraints for network-safe distributed energy resource control,'' in
  \emph{Allerton}, 2022.

\bibitem{mitzenmacher2017probability}
M.~Mitzenmacher and E.~Upfal, \emph{Probability and computing: Randomization
  and probabilistic techniques in algorithms and data analysis}.\hskip 1em plus
  0.5em minus 0.4em\relax Cambridge University Press, 2017.

\bibitem{wadsworth1961introduction}
G.~P. Wadsworth, J.~G. Bryan, and A.~C. Eringen, ``Introduction to probability
  and random variables,'' \emph{J. Appl. Mech.}, vol.~28, no.~2, p. 319, 1961.

\end{thebibliography}




\appendix

\subsection{Proof of Theorem~\ref{thm:confidence_interval_condition}}
\label{sec:thm1proof}
 \begin{proof}[\nopunct]
    By Theorem 4.1 in \cite{mitzenmacher2017probability}, the following inequality is derived from the Chernoff bound for any $0 <  \delta \leq \frac{1 - \nu}{\nu}$,
    \begin{equation} \label{eq:multiplicative_chernoff}
        \begin{aligned}
        \mathrm{Pr} ( M_{n_{\text{s}}} \geq (1+\delta) \nu ) & \leq \left ( \frac{1}{1+\delta} \right )^{(1+\delta) n_{\text{s}} \nu} e^{\delta n_{\text{s}} \nu } \\
        & = e^{ n_{\text{s}} \nu \left ( \delta - (1+\delta) \ln ( 1 + \delta ) \right ) }.
        \end{aligned}
    \end{equation}
  We substitute $c / \nu$, with $c \in [0, 1-\nu]$, for $\delta$ and obtain 
    \begin{equation} \label{eq:prob1} 
        \begin{aligned}
            & \mathrm{Pr} \left ( M_{n_{\text{s}}} - \nu \geq c \right ) \leq e^{ n_{\text{s}} \left ( c - \left ( \nu+c \right ) \ln \left ( 1 + \frac{c}{\nu} \right ) \right )} \\
            & \Longleftrightarrow \enspace \mathrm{Pr} \left ( \nu \geq M_{n_{\text{s}}} - c \right ) \geq 1 - e^{ n_{\text{s}} \left ( c- \left ( \nu +c \right ) \ln \left ( 1 + \frac{c}{\nu} \right ) \right )}.
        \end{aligned}
    \end{equation}
    Hence, $[\tilde{m}_{n_{\text{s}}}-c,1]$ is a confidence interval for $\nu$ with confidence level over $1 - e^{ n_{\text{s}} \left ( c- \left ( \nu +c \right ) \ln \left ( 1 + \frac{c}{\nu} \right ) \right )}$. Thus, if there exists $c > 0$ that satisfies $\tilde{m}_{n_{\text{s}}}-c \geq 1 - \epsilon$ and $1 - e^{ n_{\text{s}} \left ( c- \left ( \nu +c \right ) \ln \left ( 1 + \frac{c}{\nu} \right ) \right )} > 1 - \beta, $
    then $[1-\epsilon, 1]$ is a confidence interval for $\nu$ with confidence level over $1 - \beta$. Next, we show that such a $c$ exists. First, we derive a lower bound on $1 - e^{ n_{\text{s}} \left ( c- \left ( \nu +c \right ) \ln \left ( 1 + \frac{c}{\nu} \right ) \right )}$. Focusing on the exponent, observe that 
    \begin{equation} \label{eq:prob2}
    \begin{aligned}
        \frac{\partial}{\partial \nu} \left ( c - (\nu +c) \ln \left ( 1 + \frac{c}{\nu} \right ) \right ) = - \ln \left ( 1 + \frac{c}{\nu} \right ) + \frac{c}{\nu}. 
    \end{aligned}
    \end{equation}
    If we let $h_1 (x) := - \ln (1+x) + x$, the right side of \eqref{eq:prob2} is equal to $h_1 (c / \nu)$. From $h_1 (0) = 0$ and $\partial  h_1 (x)/\partial x = - 1/(1+x) + 1 \geq 0 \quad \forall x \in [0, \infty )$,
    we have $h_1(x) \geq 0$ for all  $x\in[0,\infty)$, which means $h_1 ( c/ \nu )$ is non-negative. Hence, the exponent is increasing with respect to $\nu$, and thus achieves its maximum at $\nu=1$. Therefore, 
    \begin{equation} \label{eq:prob3}
          \mathrm{Pr} \left ( \nu \geq M_{n_{\text{s}}} - c \right )  \geq 1 - e^{ n_{\text{s}} \left ( c- \left ( 1 +c \right ) \ln \left ( 1 + c \right ) \right )}.
    \end{equation}
    Since $\nu \leq 1$, \eqref{eq:prob3} implies that $[ \tilde{m}_{n_{\text{s}}} - c, 1]$ is a confidence interval for $ \nu $ with confidence level over $ 1 - e^{n_{\text{s}} \left ( c - ( 1 + c) \ln ( 1 +c) \right )}$.
    
    Now, suppose that \eqref{eq:theorem1_cond1}, \eqref{eq:theorem1_cond2} hold and define $h_2 (x) := x - (1+x) \ln (1+x)$; the exponent on the right side of \eqref{eq:prob3} is $n_{\text{s}} h_2(c)$. From $h_2(0) = 0$ and ${\partial h_2 (x)}/{\partial x}<0$ for all $x\in ( 0, \infty )$, we have
     $h_2(x)<0$ for all $x \in ( 0, \infty )$.
     Since $\tilde{m}_{n_{\text{s}}} - ( 1 - \epsilon)>0 $ by \eqref{eq:theorem1_cond1}, $\left ( \tilde{m}_{n_{\text{s}}}+ \epsilon - 1) - ( \tilde{m}_{n_{\text{s}}} + \epsilon \right ) \ln ( \tilde{m}_{n_{\text{s}}} + \epsilon) ) = h_2 ( \tilde{m}_{n_{\text{s}}}- (1-\epsilon))$ is negative. Also, substituting $c$ with $ \tilde{m}_{n_{\text{s}}} - (1- \epsilon) $, the right side of \eqref{eq:prob3} becomes
                 $1-e^{ n_{\text{s}} \left ( ( \tilde{m}_{n_{\text{s}}}+ \epsilon - 1)  -  ( \tilde{m}_{n_{\text{s}}} + \epsilon \right ) \ln ( \tilde{m}_{n_{\text{s}}} + \epsilon )  )  }$ which is less than $
                 1- e^{- \ln \left ( \frac{1}{\beta} \right ) } = 1-\beta,  $ per~\eqref{eq:theorem1_cond2}.
        Hence, $1 - e^{n_{\text{s}} \left ( c - ( 1 + c) \ln ( 1 +c) \right )} \geq 1 - \beta$ and, thus, the interval $ [1-\epsilon,1] = [ \tilde{m}_{n_{\text{s}}} -c,1]$ is a confidence interval for $\nu$ with confidence level over $ 1 - \beta$.
    \end{proof}
    
\subsection{Proof of Theorem~\ref{thm:prob_monotonicity} and supporting lemmas}\label{sec:proof_monotonicity}
We first introduce and prove Lemma~\ref{lemma:sequence}, which is required for the proof of Lemma~\ref{lem:decreasing_func}. Then, we prove Lemma~\ref{lem:decreasing_func}, which is used in the proof of Theorem~\ref{thm:prob_monotonicity}. Finally, we prove Theorem~\ref{thm:prob_monotonicity}.

\begin{lemma} \label{lemma:sequence}
     Suppose that $a_w (x), b_w (x) : \mathcal{X} \to \mathbb{R}^{+}$ are non-negative functions with parameter $w \in \mathbb{R}$, and $\{  \tilde{x}_1, \ldots, \tilde{x}_{N} \}$ ($\tilde{x}_1 \leq \ldots \leq \tilde{x}_{N}$) is a finite subset of the domain $\mathcal{X}$. Also, assume that the following two conditions hold: 1)~$\sum_{k = 1}^{j} a_{w} ( \tilde{x}_k ) $ is a decreasing function with respect to $w$ for any $j \in \{ 1 , \ldots, N \}$, and 2)~$b_w (x)$ is decreasing function with respect to both $x$ and $w$. Then, $ g (w) := \sum_{k = 1}^{N} a_w ( \tilde{x}_k ) b_w ( \tilde{x}_k ) $ is a decreasing function with respect to $w$. 
\end{lemma}

\begin{proof}
 We prove the lemma by showing that, for $w \leq \overline{w}$,  $\sum_{k = 1}^{j} a_{w} ( \tilde{x}_k ) b_{w}( \tilde{x}_k ) \geq \sum_{k =1}^{j} a_{\overline{w}} ( \tilde{x}_k ) b_{\overline{w}}( \tilde{x}_k ) $ for any $j \in [N]$ and $ w_1 , w_2 \in \mathbb{R}$ as follows:
\begin{subequations}
\allowdisplaybreaks
\label{eq:lem1_1}
\begin{align}
        & \sum_{k=1}^{j} a_{w}   ( \tilde{x}_{k} ) b_{w} ( \tilde{x}_{k} )  \geq \sum_{k=1}^{j} a_{w} ( \tilde{x}_{k} ) b_{\overline{w}} ( \tilde{x}_{k} ) \label{subeq:lem1_11} \\
         & = b_{\overline{w}} ( \tilde{x}_{j} )  \sum_{k=1}^{j} a_{w} ( \tilde{x}_{k} ) + \sum_{k=1}^{j-1}  \Delta b_{\overline{w}} (\tilde{x}_k)
         \sum_{l=1}^{k} a_{w} ( \tilde{x}_{l} ) \label{subeq:lem1_12} \\
        & \geq b_{\overline{w}} ( \tilde{x}_{j} )  \sum_{k=1}^{j}  a_{\overline{w}} (  \tilde{x}_{k} )  + \sum_{k=1}^{j-1}  \Delta b_{\overline{w}} (\tilde{x}_k)   \sum_{l=1}^{k} a_{\overline{w}} ( \tilde{x}_{l} ) \label{subeq:lem1_13} \\
         & = \sum_{k=1}^{j}  a_{\overline{w}} ( \tilde{x}_{k} )  b_{\overline{w}} (   \tilde{x}_{k}  ) 
\end{align}
\end{subequations}
where $\Delta b_{\overline{w}} (\tilde{x}_k) := ( b_{\overline{w}}  ( \tilde{x}_{k}  )  - b_{\overline{w}} ( \tilde{x}_{k+1} ) )$, \eqref{subeq:lem1_11} holds by condition~2 and \eqref{subeq:lem1_13} holds by condition~1.
\end{proof}

\begin{lemma} \label{lem:decreasing_func}
     Suppose that $Y_{w}^{(j)} $ ($j \in [ n ]$) is a discrete random variable with the finite sample space $ \mathcal{Y}^{(j)} = \{ \acute{y}_{1}^{j} ,\ldots, \acute{y}_{ \kappa_j }^{j} \}$ ($\acute{y}_{1}^{j} \leq \ldots \leq \acute{y}_{\kappa_j }^{j}$) with parameter $w \in \mathbb{R}$ having the following properties: 1) $ Y_{w}^{(1)}, \cdots, Y_{w}^{(n)} $ are independent of each other, and 2) the cdf $F_{Y^{(j)}} (y ; w)$ of $Y_{w}^{(j)}$ is a decreasing function with respect to $w$ for any $y \in \mathcal{Y}^{(j)}$.
     Then, for any $\overline{z}^{(i)} \in \mathbb{R}$ ($i \in [ n_{c} ]$) and non-negative coefficients $a_{ij} \in \mathbb{R}^{+}$, $\mathrm{Pr} \left ( \bigwedge_{i=1}^{ n_c }  \left ( \sum_{j=1}^n a_{ij} Y_w^{(j)} \leq \overline{z}^{(i)} \right ) \right ) $ monotonically decreases as $w$ increases.
\end{lemma}

\begin{proof}
    Let $\bm{Y}_{w} = ( Y_{w}^{(1)}, \ldots, Y_{w}^{(n)} )^\top $ be a multivariate random variable with elements $Y_{w}^{(j)}$ and $\mathcal{P} = \{ \bm{y} \; | \; A \bm{y} \leq \overline{\bm{z}} \}$ be a polyhedron with elements $a_{ij}$. Then, 
     \begin{equation*}
         \mathrm{Pr} \left ( \bigwedge_{i=1}^{ n_c }  \left ( \sum_{j=1}^n a_{ij} Y_w^{(j)} \leq \overline{z}^{(i)} \right ) \right ) = \mathrm{Pr} \left ( \bm{Y}_{w} \in \mathcal{P}  \right ).
     \end{equation*}
 Note that $\mathcal{P}$ is a lower polyhedron in $\Pi_{j=1}^{n} [ \acute{y}_{1}^{j}, \acute{y}_{\kappa_j}^{j} ] $; if $\bm{y} \in \mathcal{P}$, then $\bm{y}' \in \mathcal{P}$ also holds for any $\bm{y}' \leq \bm{y}$. Thus, it is sufficient to show that $\mathrm{Pr} \left ( \bm{Y}_{w_{1}} \in \mathcal{P}' \right ) \geq \mathrm{Pr} \left ( \bm{Y}_{w_{2}} \in \mathcal{P}' \right ) \, \forall \, w_1 \geq w_2$ and any lower polyhedron $\mathcal{P}'$, which we do as follows:
 \begin{enumerate}
     \item Let $n=1$ and $\mathcal{P}_{1}' \subset [\acute{y}_{1}^{1} , \acute{y}_{\kappa_1 }^{1}]$ be a 1-dimensional lower polyhedron. Then, there exists $\overline{y}$ such that $\mathcal{P}_{1}' = [ \acute{y}_{1}^{1}, \overline{y}] $, and 
             $ \mathrm{Pr} ( Y_{w_1}^{(1)} \in \mathcal{P}_{1}' ) = F_{Y^{(1)}}  ( \overline{y}  ; w_1)  \geq F_{Y^{(1)}} ( \overline{y} ; w_2 ) =  \mathrm{Pr} ( Y_{w_2}^{(1)} \in \mathcal{P}_{1}' ), $
     which proves the statement for $n=1$.
    \item Let $n=k$ and suppose $\mathrm{Pr} ( \bm{Y}_{w_1}^{(1:k)} \in \mathcal{P}_{k}' ) \geq \mathrm{Pr} ( \bm{Y}_{w_2}^{(1:k)} \in \mathcal{P}_{k}' )  $ holds $\forall \, w_1 \geq w_2$ and for any $k$-dimensional lower polyhedron $\mathcal{P}_{k}' \subset \Pi_{j=1}^{k}  [ \acute{y}_{1}^{j}, \acute{y}_{\kappa_j}^{j} ] $. Define
    $\mathcal{P}_{k}^{-} ( y_{k+1} ) = \{ (y_{1},\ldots,y_{k} )^\top \; | \;  (y_{1},\ldots,y_{k}, y_{k+1} )^\top \in \mathcal{P}_{k+1}' \}$. 
    Then, $\mathcal{P}_{k}^{-} ( y_{k+1} ) $ is a lower polyhedron for any $y_{k+1} \in [ \acute{y}_{1}^{k+1}, \acute{y}_{\kappa_{k+1}}^{k+1} ] $. 
    Therefore, $\mathrm{Pr} (  \bm{Y}_{w_1}^{(1:k+1)} \in \mathcal{P}_{k+1}' ) = \sum_{j = 1}^{\kappa_{k+1}} \mathrm{Pr} ( Y_{w_1}^{(k+1)} = \acute{y}_{j}^{k+1} ) \mathrm{Pr} ( \bm{Y}_{w_1}^{(1:k)} \in \mathcal{P}_{k}^{-} ( \acute{y}_{j}^{k+1}  ) )$ for any $k+1$-dimensional lower  polyhedron $\mathcal{P}_{k+1}' \subset \Pi_{j=1}^{k+1}  [ \acute{y}_{1}^{j}, \acute{y}_{\kappa_j}^{j} ] $. This is greater than or equal to $
            \sum_{j = 1}^{\kappa_{k+1}} \mathrm{Pr} ( Y_{w_2}^{(k+1)} = \acute{y}_{j}^{k+1} ) \mathrm{Pr} ( \bm{Y}_{w_2}^{(1:k)} \in \mathcal{P}_{k}^{-} ( \acute{y}_{j}^{k+1}  ) ) $ by Lemma~\ref{lemma:sequence}, which in turn equals $
            \mathrm{Pr} ( \bm{Y}_{w_2}^{(1:k+1)} \in \mathcal{P}_{k+1}' )$. This proves the statement  for $n=k+1$.
 \end{enumerate}
Therefore, by mathematical induction, $\mathrm{Pr} \left ( \bm{Y}_{w_{1}} \in \mathcal{P}' \right ) \geq \mathrm{Pr} \left ( \bm{Y}_{w_{2}} \in \mathcal{P}' \right )$ holds for any lower polyhedron $\mathcal{P}'$.
\end{proof}

\begin{proof} [Proof of Theorem~\ref{thm:prob_monotonicity}]
    From \eqref{eq:prop1_eq1}, we obtain
    \begin{equation}
        \begin{aligned}
        & \hat{V}_{u,j}^2 (t+1) = \hat{f}_{v_j}^{2} (\bm{P}_u (t+1), \bm{Q}_{u} (t+1) , v_0) \\
        & =  v_0^2 - 2 \sum_{k \in a(j) } \bigg ( r_k \hat{f}_{ p_{k}^{\text{b}} } ( \bm{P}_{u} (t+1), \bm{Q}_{u} (t+1), v_0 )  \\
            & \qquad \qquad \qquad \qquad + x_k \hat{f}_{q_{k}^{\text{b}}} ( \bm{P}_{u} (t+1), \bm{Q}_{u} (t+1), v_0 ) \bigg ) \\
        & =  v_0^2 - 2 \sum_{k \in a(j) }  \sum_{l \in d(k) } \left ( r_k  P_{u,l} (t+1) + x_k Q_{u,l} (t+1) \right ).
        \end{aligned}
        \label{eq:v2}
    \end{equation}
    Substituting $P_{u,l} (t+1)$ with $P_{l}^{\text{L}} (t+1) + \overline{p}_{l} N_{u,l}^{\text{ON}} (t+1)$, $Q_{u,l} (t+1)$ with $Q_{l}^{\text{L}} (t+1) + \overline{q}_{l} N_{u,l}^{\text{ON}} (t+1)$, $N_{u,l}^{\text{ON}} (t+1)$ with the right side of \eqref{eq:nextnOn}, and leveraging \eqref{eq:v2} we obtain
    \begin{equation*}
        \begin{aligned}
             \hat{V}_{u,j} (t+1) \geq \underline{v} &\Longleftrightarrow \hat{V}_{u,j}^2 (t+1) \geq \underline{v}^2\\
            & \Longleftrightarrow g_j ( \bm{C}_{u} (t+1)) \leq h_j (\bm{R} ),
        \end{aligned}
    \end{equation*}
    where vector $\bm{R} := (\bm{N}^{\text{ON}} (t)^\top, \bm{P}^{\text{L}} (t+1)^\top,\bm{Q}^{\text{L}} (t+1)^\top, \bm{S}^{\text{ON}} (t+1)^\top ,\bm{S}^{\text{OFF}} (t+1)^\top )^\top$ collects random variables, $\bm{C}_{u} (t+1) := \bm{C}_{u}^{\text{ON}} (t+1) - \bm{C}_{u}^{\text{OFF}} (t+1)$ is the net number of TCL OFF to ON switches by the aggregator's command, and the functions $g_j$ and $h_j$ are 
    \begin{equation*}
        \begin{aligned}
            & g_j \left ( \bm{C}_{u} (t+1) \right ) :=  2 \sum_{k \in a (j)} \sum_{l \in d(k)} (r_k \overline{p}_l + x_k \overline{q}_{l}) C_{u,l} (t+1), \\
            & h_j (\bm{R}) :=  v_0^2 - \underline{v}^2 - 2 \sum_{k \in a(j) } \sum_{l \in d(k) }  \Big ( r_k  P_{l}^{\text{L}} (t+1)+  x_k Q_{l}^{\text{L}} (t+1) \\
            & \qquad  + (r_k \overline{p}_l + x_k \overline{q}_{l} ) \left (  N_{l}^{\text{ON}} (t)  +  S_{l}^{\text{ON}} (t+1) - S_{l}^{\text{OFF}} (t+1) \right ) \Big ).
        \end{aligned}
    \end{equation*}
    Note that $g_j$ is a non-negative linear combination of $C_{u,l} (t+1)$ for all $j \in [n]$, i.e.,  there exist $a_{jl} \geq 0$ for any $j,l \in [n]$ such that $g_j ( \bm{C}_{u} (t+1) )$ is equal to $\sum_{l=1}^n a_{jl} C_{u,l} (t+1) $.
    
    Let $\mathcal{R}$ be the sample space of $\bm{R}$ and $f_{\bm{R}}$ be the joint probability density function of $\bm{R}$. Then, we have
   \begin{equation} \label{eq:thm3proof_1}
    \begin{aligned}
        & \hat{\nu}_{u} (t+1) = \mathrm{Pr} \left ( \bigwedge_{j=1}^{n} \left ( \hat{V}_{u,j} (t+1) \geq \underline{v} \right )  \right ) \\
        & = \int_{\tilde{r} \in \mathcal{R}}\mathrm{Pr} \left ( \bigwedge_{j=1}^{n} \left ( g_j ( \bm{C}_{u} (t+1)) \leq h_j ( \tilde{r}) \right ) \bigg | \bm{R} = \tilde{r} \right ) f_{\bm{R}} (\tilde{r}) d \tilde{r}.
    \end{aligned}
    \end{equation}
    
   For any realization, $\tilde{r} := (\tilde{\bm{n}}^{\text{ON}} (t)^\top,\tilde{\bm{p}}^{\text{L}} (t+1)^\top , \tilde{\bm{q}}^{\text{L}} (t+1)^\top, \tilde{\bm{s}}^{\text{ON}} (t+1)^\top, \tilde{\bm{s}}^{\text{OFF}}
    (t+1)^\top )^\top \in \mathcal{R}$,  $C_{u,l} (t+1) = C_{u}^{\text{ON}} (t+1) $ when $u\geq 0$, and $C_{u,l} ( t+1) = -C_{u}^{\text{OFF}} ( t+1)$ when $u < 0$. Thus, by \eqref{eq:numofswitchTCLs}, 
    the conditional cdf of $C_{u,l} ( t+1)$ is computed as $\mathrm{Pr} ( C_{u,l} (t+1) \leq k | \bm{R} = \tilde{r} )= \mathcal{F}_{\text{B}} (k ; \tilde{n}_{l}^{\text{OFF}} (t) - \tilde{s}_{l}^{\text{ON}} (t+1), u)$ when $u \geq 0$, and $\mathrm{Pr} ( C_{u,l} (t+1) \leq k | \bm{R} = \tilde{r} ) =  1 - \mathcal{F}_{\text{B}} (-k ; \tilde{n}_{l}^{\text{ON}}(t) - \tilde{s}_{l}^{\text{OFF}}(t+1) , -u)$ when $u<0$. In addition, from \cite{wadsworth1961introduction}, the cdf of a binomial random variable $\mathcal{B} (n;\nu)$ is
    \begin{equation*}
        \mathcal{F}_{\text{B}} (k ; n,\nu) = (n-k) \binom{n}{k} \int_{0}^{1-\nu} t^{n-k-1} (1-t)^k dt,
    \end{equation*}
    which is a monotonically decreasing function with respect to $\nu$. Thus, $\mathrm{Pr} ( C_{u,l} (t+1) \leq k | \bm{R} = \tilde{r} )$  monotonically decreases as $u$ increases, and $ C_{u,1} (t+1)|\tilde{r}, \ldots,  C_{u,n}(t+1)|\tilde{r} $ for any $\tilde{r} \in \mathcal{R}$ satisfies the conditions on the random variables in Lemma~\ref{lem:decreasing_func}. Thus,  
    $
        \mathrm{Pr} \left ( \bigwedge_{j=1}^{n} \left ( g_j ( \bm{C}_{u} (t+1)) \leq h_j ( \tilde{r} ) \bigg | \bm{R} = \tilde{r} \right ) \right )
    $
    is a decreasing function with respect to $u$. Therefore, by \eqref{eq:thm3proof_1},  $\hat{\nu}_{u} (t+1)$ is also a decreasing function with respect to $u$.
\end{proof}

\end{document}